\documentclass[11pt]{article}
\usepackage{amssymb}

\usepackage{amsmath}

\newcounter{resultnum}[section]

\newcounter{conclusionnum}[section]

\newcounter{conditionnum}[section]

\newcounter{conjecturenum}[section]

\newcounter{examplenum}[section]

\newcounter{exercisenum}[section]
\newtheorem{lemma}{Lemma}[section]

\newcounter{lemmanum}[section]

\newcounter{notationnum}[section]
\newtheorem{theorem}{Theorem}[section]

\newcounter{theoremnum}[section]
\newtheorem{definition}{Definition}[section]

\newcounter{definitionnum}[section]
\newtheorem{corollary}{Corollary}[section]

\newcounter{corollarynum}[section]

\newcounter{remarknum}[section]
\newtheorem{proposition}{Proposition}[section]

\newcounter{propositionnum}[section]

\newcounter{acknowledgementnum}[section]

\newcounter{algorithmnum}[section]

\newcounter{axiomnum}[section]

\newcounter{casenum}[section]

\newcounter{claimnum}[section]

\newcounter{summarynum}[section]

\newcounter{problemnum}[section]
\newenvironment{proof}[1][]{\textbf{Proof.} }{}

\begin{document}
\title{\textbf{Fedosov Quantization of Lagrange--Finsler and
Hamilton--Cartan Spaces and Einstein Gravity Lifts on (Co) Tangent Bundles}}
\date{November 9, 2008}
\author{\textbf{Mihai Anastasiei}\thanks{anastas@uaic.ro} \\
\textsl{{\small Faculty of Mathematics, University "Al. I. Cuza" Ia\c{s}i,}}
\\
\textsl{{\small 11, Carol I Boulevard, Ia\c{s}i, Romania, 700506 }} \\
\textsl{\small and } \\
\textsl{{\small Mathematical Institute "O. Mayer", Romanian Academy Ia\c{s}i
Branch,}} \\
\textsl{{\small 8, Carol I Boulevard, Ia\c{s}i, Romania, 700506 } } \and
\textbf{Sergiu I. Vacaru}\thanks{Sergiu.Vacaru@gmail.com ;\
http://www.scribd.com/people/view/1455460-sergiu} \\
\textsl{\small The Fields Institute for Research in Mathematical Science} \\
\textsl{\small 222 College Street, 2d Floor, Toronto \ M5T 3J1, Canada} \\
\textsl{\small and } \\
\textsl{{\small Faculty of Mathematics, University "Al. I. Cuza" Ia\c{s}i,}}
\\
\textsl{{\small 11, Carol I Boulevard, Ia\c{s}i, Romania, 700506 }}
}
\maketitle

\begin{abstract}
We provide a method of converting Lagrange and Finsler spaces and
their Legendre transforms to Hamilton and Cartan spaces into
almost K\"{a}hler structures on tangent and cotangent bundles. In
particular cases, the Hamilton spaces contain nonholonomic lifts
of (pseudo) Riemannian / Einstein metrics on effective phase
spaces. This allows us to define the corresponding Fedosov
operators and develop deformation quantization schemes for
nonlinear mechanical and gravity models on Lagrange-- and
Hamilton--Fedosov manifolds.


\vskip3pt


\end{abstract}


\section{ Introduction}

To construct a quantum theory for a given classical model is
usually understood that it is necessary to elaborate a
quantization procedure adapted to certain fundamental field
equations and corresponding Lagrangians/ Hamiltonians and theirs
symmetries, constraints and locality. In various attempts to
develop quantum versions of gravity and nonlinear physical
theories, one provided different computation schemes when it is
supposed that all constraints can be solved, gauge symmetries can
be represented by shift symmetries and there are certain canonical
forms for the Poisson brackets. Nevertheless, such quantization
usually destroy global symmetries, locality and may result into
violation of local Lorentz symmetry.

There were proposed different sophisticate constructions with formal and
partial solutions for quantum gravity and field interactions theories. We
cite here the BRST quantization methods for non--Abelian and open gauge
algebras \cite{bff,wholt,grl,henn}, deformation quantization \cite%
{fed1,fed2,konts1,konts2}, quantization of general Lagrange structures and,
in general, BRST quantization without Lagrangians and Hamiltonians \cite%
{lyakh1,lyakh2}, $W$--geometry and Moyal deformations of gravity via strings
and branes \cite{castro1,castro2,castro3} and quantum loops and spin
networks \cite{rovelli,asht3,thiem1}.

In the above--mentioned approaches, it is necessary to quantize curved phase
spaces in a manner that is explicitly covariant on phase--space coordinates.
Indeed, for instance, the key ingredients of Fedosov and BRST methods, see %
\cite{lyakh1,lyakh2,gls,bgl} and references therein, is the embedding of the
system into the cotangent bundle over its phase space. There were also
elaborated such models following standard geometric constructions over
Riemannian manifolds and formal schemes with symplectic groupoids and
contravariant connections \cite{gls,karabeg}.

A rigorous geometric approach to deformation quantization of gravity, gauge
theories and geometric mechanics models with constraints and related
generalized Lagrange--Finsler theories, see \cite%
{esv,vqgr1,vqgr2,vqgr3,vqgr4}, shows that the quantization schemes have to
be developed for nonholonomic manifolds\footnote{%
i.e. manifolds endowed with nonholonomic (equivalent, anholonomic,
or non-integrable distributions), see details and references in
\cite {vr1,vr1a,vr2,bejf,vsgg,vrfg}} and tangent and cotangent
bundles endowed with nonlinear connection (N--connection)
structure. The natural step in this direction is to apply the
methods of the geometry of Hamilton and Cartan spaces and
generalizations \cite{mhss,mhh} (such spaces are
respectively dual to the Lagrange and Finsler spaces and generalizations %
\cite{ma1987,ma}).

The aim of this work is to show how Karabegov's approach to Fedosov
deformation quantization \cite{karabeg,karabeg1,karabeg2} can be naturally
extended for almost K\"{a}hler manifolds endowed with canonical geometric
structures generated by semi--Riemannian and/or Einstein metrics and
Lagrange--Finsler and Hamilton--Cartan fundamental generating functions.

This paper is motivated by the following results: In Refs. \cite%
{vrfg,vqgr1,vqgr2,vstav,vsgg}, we concluded that classical and quantum
gravity models on (co) tangent bundles positively result in generalized
Finsler like theories with violation of local Lorentz symmetry. The
conclusion was supported also by a series of works on definition of spinors
and field interactions on (in general, higher order) locally anisotropic
spacetimes \cite{vfs,vhs}, on low energy limits of (super) string theory %
\cite{vstrf,vncsup,mavr} and possible Finsler like phenomenological
implications and symmetry restriction of quantum gravity \cite{glsf,mign,ggp}%
. Here, we emphasize that the nonholonomic quantum deformation formalism can
be re--defined for nonholonomic (pseudo) Riemannian, or Riemann--Cartan,
manifolds with fibred structure. Such nonholonomic spaces, under well
defined conditions and for corresponding Lagrange--Finsler variables,
possess local Lorentz invariance for classical theories and seem to preserve
it for nonholonomic quantum deformations, see results from \cite{vqgr3,vqgr4}%
.

The work is organized as follows: In Section 2 we recall the basis
of the Lagrange--Finsler and Hamilton--Cartan geometry,
nonholonomic lifts of Einstein metrics on cotangent bundles and
almost K\"{a}hler models of such spaces. We also collect there
some geometric constructions that we need for further
considerations. Section 3 is devoted to the formalism of canonical
nonlinear connections and distinguished connections induced by
Lagrange and Hamilton fundamental functions. We introduce the
concept of Hamilton--Fedosov spaces and define the corresponding
almost symplectic structures. In Section 4 we consider a
generalization of the concept of connection to that of
connection--pair resulting in Fedosov--Hamilton operators--pairs
which is necessary for definition of deformation quantization
models being invariant under symplectic morphisms and Legendre
transform. We provide  Fedosov's theorems for connection--pairs
which allows us to develop an approach to geometric quantization
of Hamilton and Einstein (generalized on cotangent bundles) spaces
in Section 5. We speculate on possible quantum generalizations of
gravitational field equations on phase spaces and their
deformation quantization. Finally, in Section 6 we conclude the
results.

\section{Lagrange--Finsler and Hamilton--Cartan Geometry and Einstein Spaces}

In this section, we outline some results from the geometry of
Lagrange--Finsler \cite{ma1987,ma} and Hamilton--Cartan
\cite{mhss,mhh} spaces.

\subsection{Canonical geometric objects on Lagrange and Hamilton spa\-ces}

Let us consider a real sufficiently smooth manifold $M, \dim M=n\geq
2.$ We label the local coordinates $x=\{x^{i}\},$ with ''base'' indices $%
i,j,...=1,2,...n,$ and write by $TM$ and $T^{\ast }M,$ respectively, the
total spaces of tangent and cotangent bundles $\pi :TM\rightarrow M$ and $%
\pi ^{\ast }:T^{\ast }M\rightarrow M,$ with local coordinates $%
u=(x,y)=\{u^{\alpha }=(x^{i},y^{a})\},$ for ''fiber'' indices $%
a,b,...=n+1,...,n+n$ on any local cart $U\subset TM,$ and $\ ^{\ast
}u=(x,p)=\{\ ^{\ast }u^{\alpha }=(x^{i},p_{a})\}$ on any local cart $\
^{\ast }U\subset T^{\ast }M,$ when $p_{a}$ are dual to $y^{a}.$ In order to
apply the Einstein summation rule for contracting base and fiber indices, we
shall use identifications of type $y^{i}\doteqdot y^{n+i},$ i.e. we suppose
that indices like $a,b,...$ can split into respective $n+i,n+j.$

\begin{definition}
A generalized Hamilton space is defined by a pair
$GH^{n}=(M,g^{ij}(x,p)),$ where $g^{ij}(x,p)$ is a contravariant
symmetric tensor field, non-degenerate
and of constant signature on $\widetilde{T^{\ast }M}=T^{\ast }M/\{0\},$ for $%
\{0\}$ being the null section of $T^{\ast }M.$
\end{definition}

We note that a contravariant tensor of type
\begin{equation}
g^{ij}(x,p)=e_{\ i^{\prime }}^{i}(x,p)e_{\ j^{\prime }}^{j}(x,p)g^{i^{\prime
}j^{\prime }}(x)  \label{eq01}
\end{equation}%
includes, for some vielbein fields $e_{\ i^{\prime }}^{i},$ the dual tensor
\ $g^{i^{\prime }j^{\prime }}(x)$ as the inverse of a (semi) Riemann metric $%
g_{i^{\prime }j^{\prime }}(x)$ on $M.$ In this paper, we consider
that any classical solution of the Einstein equations defines a
space-time manifold $M$ and a corresponding deformation
quantization procedure on $T^{\ast }M$ defines the values $e_{\
i^{\prime }}^{i}$ and other fundamental quantum geometric objects
(like almost symplectic structure and generalized connection) with
nontrivial dependence on variables $p_{a}.$ In general, we shall
provide our constructions for (pseudo) Riemannian spaces with (co)
metrics of type $g^{i^{\prime }j^{\prime }}(x)$ and $g^{ij}(x,p)$
and discuss if there are any important particular properties for
Einstein manifolds (spaces) when $g_{i^{\prime }j^{\prime }}(x)$
is subjected to the condition to solve on $M$ the gravitational
field equations with nonzero cosmological constant $\lambda ,$
\begin{equation}
\ _{\shortmid }R_{i^{\prime }j^{\prime }}(x)=\lambda g_{i^{\prime }j^{\prime
}}(x),  \label{einstm}
\end{equation}%
where $\ _{\shortmid }R_{i^{\prime }j^{\prime }}(x)$ is the Ricci tensor for
the Levi--Civita connection $\nabla =\{\ _{\shortmid }\Gamma _{\ j^{\prime
}k^{\prime }}^{i^{\prime }}(x)\}$ completely defined by $g_{i^{\prime
}j^{\prime }}.$

For simplicity, we shall work with a more particular class of spaces when $%
g^{ij}(x,p)$ is defined by a Hamilton function $H(x,p):$

\begin{definition}
A Hamilton space $H^{n}=(M,H(x,p))$ is defined by a function
$T^{\ast }M\ni (x,p)\rightarrow H(x,p)\in \mathbb{R},$ i.e. by
fundamental Hamilton function, $\ $which is differentiable on
$\widetilde{T^{\ast }M}$ and continuous on the null section $\pi
^{\ast }:T^{\ast }M\rightarrow M$ and such that the (Hessian,
equivalently, fundamental) tensor field
\begin{equation}
\ ^{\ast }g^{ab}(x,p)=\frac{\partial ^{2}H}{\partial p_{a}\partial p_{b}}
\label{hm}
\end{equation}
is non-degenerate and of constant signature on $\widetilde{T^{\ast }M}.$
\end{definition}

Let $L(x,y)$ be a regular differentiable Lagrangian on $U\subset TM,$ with
non-degenerate Hessian (equivalently, fundamental tensor field)
\begin{equation}
g_{ab}(x,y)=\frac{\partial ^{2}L}{\partial y^{a}\partial y^{b}}.  \label{lm}
\end{equation}

\begin{definition}
A Lagrange space $L^{n}=(M,L(x,y))$ is defined by a function $TM\ni
(x,y)\rightarrow L(x,p)\in \mathbb{R},$ i.e. a fundamental Lagrange
function, which is differentiable on $\widetilde{TM}$ and continuous on
the null section of $\pi :TM\rightarrow M$ and such that the (Hessian) tensor field $%
g_{ab}(x,y)$ (\ref{lm}) is non-degenerate and of constant
signature on $\widetilde{TM}.$
\end{definition}

We can define the Legendre transform $L\rightarrow H,$%
\begin{equation}
H(x,p)=p_{a}y^{a}-L(x,y),  \label{legt1}
\end{equation}%
where $y=\{y^{a}\}$ are solutions of the equations $p_{a}=\partial
L(x,y)/\partial y^{a},$ and (inversely) the Legendre transform $H\rightarrow
L,$%
\begin{equation}
L(x,y)=p_{a}y^{a}-H(x,p),  \label{legt2}
\end{equation}
where $p=\{p_{a}\}$ is the solution of the equations $y^{a}=\partial
H(x,p)/\partial p_{a}.$ \footnote{%
In some monographs (for instance, see \cite{ma,mhss}), it is considered the
factor $1/2$ in the right sides of (\ref{hm}) and (\ref{lm}). We emphasize
that in this paper (for simplicity) the Hamilton and Lagrange functions will
be supposed to be regular and related mutually by Legendre transforms.}

Following terminology from \cite{mhss}, we say that $\ N_{i}^{a}$ (\ref{clnc}
) and $\ ^{\ast }N_{ij}$ (\ref{chnc}) are $\mathcal{L}$--dual if $L$ and $H$
are related by Legendre transform (\ref{legt1}), or (\ref{legt2}). In the
following constructions, we shall consider that to Legendre transform there
are associated the diffeomorphisms
\begin{equation*}
\varphi :TM\supset U\rightarrow \ ^{\ast }U\subset T^{\ast }M,
(x^i,y^a)\rightarrow \left( x^i, p_a = \frac{\partial L(x,y)}{\partial y^a} \right)
\end{equation*}
and
\begin{equation*}
\psi :T^{\ast }M\supset \ ^{\ast }U\rightarrow U\subset TM,
(x^i,p_a)\rightarrow \left(x^i,y^a = \frac{\partial H(x,p)}{\partial p_a}\right),
\end{equation*}
allowing to define respectively pull--back and push--forward of geometric
objects (functions, vectors, differential forms, connections, tensors...)
from $\ ^{\ast }U$ to $U$ and from $\ U$ to $\ ^{\ast }U,$ i.e. we define $%
\mathcal{L}$--dual geometric objects. For instance, for a differentiable
function $\ _{1}f$ on $U,$ we define a differentiable function $\
_{1}f^{\ast }\doteqdot \ _{1}f\circ \psi =\ _{1}f\circ \varphi ^{-1}$ on $\
^{\ast }U$ and (inversely) for a differentiable function $\ ^{2}f$ on $\
^{\ast }U,$ we have a differentiable $\left( \ ^{2}f\right) ^{0}\doteqdot \
^{2}f\circ \varphi =\ ^{2}f\circ \psi ^{-1}$ on $U.$ Note that $H =L\circ
\varphi^{-1}$ and $L =H\circ \psi^{-1}.$ Similarly, for any vector field $X$ on $%
U, $ we get a vector field $\ ^{\ast }X\doteqdot T\varphi \circ
X\circ \varphi ^{-1}=T\psi ^{-1}\circ X\circ \psi $ on $\ ^{\ast
}U$ and (inversely) for a vector field $\ ^{\ast }X$ on $\ ^{\ast
}U,$ we get a vector field $\ ^{\circ }X\doteqdot T\psi \circ
X\circ \psi ^{-1}=T\varphi ^{-1}\circ \ ^{\ast }X\circ \varphi $
on $U,$ where, for example, $T\varphi $ is the tangent map to
$\varphi .$ Dualizing the vector constructions, we obtain that for
any 1--form $\omega $ on $U,$ there is 1-form $\ ^{\ast }\omega
\doteqdot (T\varphi )^{\ast }\circ \omega \circ \varphi
^{-1}=(T\psi ^{-1})^{\ast }\circ \omega \circ \psi $ on $\ ^{\ast
}U$ and (inversely) for any form $\ ^{\ast }\omega $ on $\ ^{\ast
}U,$ we can consider $\ ^{\circ }\omega \doteqdot (T\psi )^{\ast
}\circ \ ^{\ast }\omega \circ \psi ^{-1}=(T\varphi
^{-1})\circ \ ^{\ast }\omega \circ \varphi $ on $U,$ where, for example, $%
(T\psi )^{\ast }$ denotes the cotangent map to $(T\psi ).$

Let $vTM$ and $vT^*M$ be the vertical distributions on $TM$ and
$T^*M$, respectively.

\begin{definition}
\label{defnc}Any Whitney sums
\begin{equation}
TTM=hTM\oplus vTM  \label{whitney}
\end{equation}%
and
\begin{equation}
TT^{\ast }M=hT^{\ast }M\oplus vT^{\ast }M  \label{whitneyd}
\end{equation}%
define respectively nonlinear connection (N--connection) structures
paramet\-riz\-ed by the local vector fields
\begin{equation*}
e_i = \frac{\partial}{\partial x^i} - N_{i}^{a}(x,y)\frac{\partial }{%
\partial y^{a}} \mbox{ \ on \ }TM
\end{equation*}%
and
\begin{equation*}
\ ^{\ast }e_i= \frac{\partial}{\partial x^i} + \ ^{\ast }N_{ia}(x,p)\frac{%
\partial }{\partial p_{a}}\mbox{ \ on \ }T^{\ast }M.
\end{equation*}
\end{definition}

One says that a N--connection defines on $TM,$ or $T^{\ast }M,$ a
conventional horizontal (h) and vertical (v) splitting (decomposition).

Let consider a regular curve $c(\tau )$ with real parameter $\tau ,$ when $%
c:\tau \in \lbrack 0,1]\rightarrow x^{i}(\tau )\subset U.$ It can
be lifted to $\pi ^{-1}(U)\subset \widetilde{TM}$ as
$\widetilde{c}(\tau ):\tau \in \lbrack 0,1]\rightarrow \left(
x^{i}(\tau ),y^{i}(\tau )=\frac{dx^{i}}{d\tau }\right) $ since the
vector field $\frac{dx^{i}}{d\tau }$ does not vanish on
$\widetilde{TM}.$ Following techniques from variational calculus,
one proves:

\begin{theorem}
\label{th1}The Euler--Lagrange equations,
\begin{equation}
\frac{d}{d\tau }\frac{\partial L}{\partial y^{i}}-\frac{\partial L}{\partial
x^{i}}=0,  \label{eleq}
\end{equation}%
are equivalent to the Hamilton--Jacobi equations,%
\begin{equation}
\frac{dx^{i}}{d\tau }=\frac{\partial H}{\partial p_{i}}\mbox{ and }\frac{%
dp_{i}}{d\tau }=-\frac{\partial H}{\partial x^{i}},  \label{hjeq}
\end{equation}%
and to the nonlinear geodesic (semi--spray) equations%
\begin{equation}
\frac{d^{2}x^{i}}{d\tau ^{2}}+2G^{i}(x,y)=0,  \label{ngeq}
\end{equation}%
where%
\begin{equation*}
G^{i}=\frac{1}{2}g^{ij}\left( \frac{\partial ^{2}L}{\partial y^{j}\partial
x^{k}}y^{k}-\frac{\partial L}{\partial x^{j}}\right),
\end{equation*}%
for $g^{ij}$ being the inverse to $g_{ij}$ (\ref{lm}).
\end{theorem}

Let us consider on $T^{\ast }M$ the canonical symplectic structure
\begin{equation}
\theta \doteqdot dp_{i}\wedge dx^{i}.  \label{csymps}
\end{equation}%

The Hamiltonian $H$ defines an unique vector field on $T^*M$:

\begin{equation*} X_{H}=\frac{\partial H}{\partial
p_{i}}\frac{\partial }{\partial x^{i}}- \frac{\partial H}{\partial
x^{i}}\frac{\partial }{\partial p_{i}}
\end{equation*}
by the equation
\begin{equation*} i_{X_{H}}\theta =-dH,
\end{equation*}
where $i_{X_{H}}$ denotes the interior product by $X_H$. The same
holds for any function on $T^{\ast }M.$

By Theorem \ref{th1} one has:

\begin{corollary}
\label{chje}The Hamilton--Jacobi equations (\ref{hjeq}) are equivalent to
\begin{equation*}
\frac{dx^{i}}{d\tau }=\{H,x^{i}\}\mbox{ and }\frac{dp_{a}}{d\tau }
=\{H,p_{a}\},
\end{equation*}
where the Poisson structure is defined by brackets
\begin{equation}
\{\ ^{1}f,\ ^{2}f\}=\theta (X_{\ ^{1}f},X_{\ ^{2}f})  \label{poisbr}
\end{equation}
for any functions $\ ^{1}f(x,p)$ and $\ ^{2}f(x,p)$ on $T^{\ast
}M.$
\end{corollary}

\begin{proof}
It can be obtained by  a standard calculus in geometric mechanics.
$\Box $
\end{proof}

\vskip5pt

The following theorem holds:

\begin{theorem}
\label{thcnc}The are canonical N--connections defined respectively by
regular Lagrange $L(x,y)$ and/ or Hamilton $H(x,p)$ fundamental functions:
\begin{equation}
\ N_{i}^{a}\doteqdot \frac{\partial G^{a}}{\partial y^{i}}  \label{clnc}
\end{equation}
and
\begin{equation}
\ ^{\ast }N_{ij}\doteqdot \frac{1}{2}\left[ \{\ ^{\ast}g_{ij},H\}-\frac{\partial ^{2}H}{\partial p_{k}\partial x^{(i}}\
^{\ast }g_{j)k}\right] , \label{chnc}
\end{equation}
where, $^{\ast}g_{ij}$ is the inverse to $^{\ast}g^{ij}$(\ref{hm})
and, for instance, $a_{(ij)}=a_{ij}+a_{ji}$ denotes symmetrization
of indices.
\end{theorem}

\begin{proof}
We can verify respectively that on any open sets $U\subset TM$ and
$U^{\ast }\subset T^{\ast }M$ $\ $coefficients (\ref{clnc}) and
(\ref{chnc}) satisfy the conditions of Definition \ref{defnc}. For
details see Ch. 9 in \cite{ma} and Ch. 5 in \cite{mhss}. $\Box $
\end{proof}

On (co) tangent bundles endowed with N--connection structure, it
is convenient to elaborate a covariant calculus adapted to this
structure, i.e. preserving the conventional splitting of tensors
and other geometric objects (like connections, differential forms
etc) into horizontal (h) and vertical (v) components. In brief,
such distinguished (by N--connection) components are called
respectively d--objects, d--field (for some physical fields of
tensor, spinor nature ...),  d--tensors, d--vectors, d--forms,
d--connections etc, see details in Refs.
\cite{mhss,mhh,ma1987,ma,vfs,vsgg}.

\vskip5pt

\begin{proposition}
\label{pnaf}There are canonical frame structures (local N--adapted
(co--)bases ) defined by canonical N--connections:%
\begin{eqnarray}
\mathbf{e}_{\alpha } &=&(\mathbf{e}_{i}=\frac{\partial }{\partial x^{i}}%
-N_{i}^{a}\frac{\partial }{\partial y^{a}},e_{b}=\frac{\partial }{\partial
y^{b}}),\mbox{ on }TM,  \label{dder} \\
\ ^{\ast }\mathbf{e}_{\alpha } &=&(\ ^{\ast }\mathbf{e}_{i}=\frac{\partial }{%
\partial x^{i}}+\ ^{\ast }N_{ia}\frac{\partial }{\partial p_{a}},\ ^{\ast
}e^{b}=\frac{\partial }{\partial p_{b}}),\mbox{ on }T^{\ast }M,
\label{cdder}
\end{eqnarray}%
and their dual (coframe) structures%
\begin{eqnarray}
\mathbf{e}^{\alpha } &=&(e^{i}=dx^{i},\mathbf{e}^{b}=dy^{b}+N_{i}^{b}dx^{i}),%
\mbox{ on }(TM)^{\ast },  \label{ddif} \\
\ ^{\ast }\mathbf{e}^{\alpha } &=&(\ ^{\ast }e^{i}=dx^{i},\ ^{\ast }\mathbf{e%
}_{p}=dp_{b}-\ ^{\ast }N_{ib}dx^{i}),\mbox{ on }(T^{\ast }M)^{\ast },
\label{cddif}
\end{eqnarray}%
when $\mathbf{e}_{\alpha }\rfloor \mathbf{e}^{\beta }=\delta _{\alpha
}^{\beta }$ and $\ ^{\ast }\mathbf{e}_{\alpha }\rfloor \ ^{\ast }\mathbf{e}%
^{\beta }=\delta _{\alpha }^{\beta },$ where by $\rfloor $ we note the
interior products and $\delta _{\alpha }^{\beta }$ being the Kronecker delta
symbol.
\end{proposition}

\begin{proof}
It follows by construction under the condition that such frames should
depend linearly on coefficients of respective N--connections. $\Box $
\end{proof}

\vskip5pt

One says that certain geometric objects are defined on $TM$ (or
$T^{\ast }M)$ in N--adapted form [equivalently, in distinguished
form, in brief,
d--form] if they are given by coefficients defined with respect to frames $%
\mathbf{e}_{\alpha }$ (\ref{dder}) and coframes $\mathbf{e}^{\alpha }$ (\ref%
{ddif}) and their tensor products (with respect to frames $\ ^{\ast }\mathbf{%
e}_{\alpha }$ (\ref{cdder}) and coframes $\ ^{\ast }\mathbf{e}^{\alpha }$ (%
\ref{cddif}) and their tensor products). We shall use ''boldface''
letters in order to emphasize that certain spaces (or geometric
objects) are in N--adapted form.

\begin{definition}
The N--lifts of the fundamental tensor fields $\ ^{\ast }g^{ab}$ (\ref{hm})
and $g_{ab}$ (\ref{lm}) are respectively
\begin{equation}
\ ^{\ast }\mathbf{g}=\ ^{\ast }\mathbf{g}_{\alpha \beta }\ ^{\ast }\mathbf{e}%
^{\alpha }\otimes \ ^{\ast }\mathbf{e}^{\beta }=\ ^{\ast
}g_{ij}(x,p)e^{i}\otimes e^{j}+\ ^{\ast }g^{ab}(x,p)\ ^{\ast }\mathbf{e}%
_{a}\otimes \ ^{\ast }\mathbf{e}_{b},  \label{dmctb}
\end{equation}%
on $T^{\ast }M,$ where $\ ^{\ast }g_{ij}$ is inverse to $\ ^{\ast }g^{ab},$
and
\begin{equation*}
\mathbf{g}=\mathbf{g}_{\alpha \beta }\ \mathbf{e}^{\alpha }\otimes \mathbf{e}%
^{\beta }=g_{ij}(x,y)e^{i}\otimes e^{j}+g_{ab}(x,y)\mathbf{e}^{a}\otimes
\mathbf{e}^{b},
\end{equation*}%
on $TM,$ where $g_{ij}$ is stated by $g_{ab}$ following $g_{ij}=g_{n+i\
n+j}. $
\end{definition}

The following proposition holds:

\begin{proposition}
The canonical N--connections $\mathbf{N}$ (\ref{clnc}) and $\ \ ^{\ast }%
\mathbf{N}$ (\ref{chnc}) define respectively the canonical almost complex
structures $\mathbf{J,}$ on $TM,$ and$\ \ ^{\ast }\mathbf{J,}$ on $T^{\ast
}M.$
\end{proposition}

\begin{proof}
On $TM$ one introduces the linear operator $\mathbf{J}$ acting on $\mathbf{e}%
_{\alpha }=(\mathbf{e}_{i},e_{b})$ (\ref{dder}) as follows:
\begin{equation*}
\mathbf{J}(\mathbf{e}_{i})=-\mathbf{e}_{n+i} \mbox{\ and \
}\mathbf{J}(e_{n+i})=\mathbf{e}_{i}.
\end{equation*}

It is clear that $\mathbf{J}$ defines globally an almost complex structure ($%
\mathbf{J\circ J=-I}$ for $\mathbf{I}$ being the unity matrix) on $TM$
completely determined for Lagrange spaces by a $L(x,y).$ Now we provide the
proof for $T^{\ast }M$. Let us introduce a linear operator $^{\ast }\mathbf{J%
}$ acting on $^{\ast }\mathbf{e}_{\alpha }=( ^{\ast }\mathbf{e}_{i}, ^{\ast
}e^{b})$ (\ref{cdder}) following formulas
\begin{equation*}
^{\ast }\mathbf{J}( ^{\ast }\mathbf{e}_{i})=-g_{ia} ^{\ast
}e^{n+i} \mbox{\ and \ } \mathbf{^{\ast}{J}}(^{\ast}e^{n+i})=\
^{\ast }\mathbf{e}_{i}.
\end{equation*}
Then $\ ^{\ast }\mathbf{J}$ defines globally an almost complex structure ( $%
\ ^{\ast }\mathbf{J\circ \ ^{\ast }J}=$ $-\mathbf{\ I}$ for $\mathbf{I}$
being the unity matrix) on $T^{\ast }M$ completely determined for Hamilton
spaces by a $H(x,p).$ $\Box $
\end{proof}

\vskip5pt

\begin{definition}
The Neijenhuis tensor field for the almost complex structure $\ ^{\ast }%
\mathbf{J}$ on $T^{\ast }M,$ or $\mathbf{J}$ on $TM,$ defined by a
N--connection (equivalently, the curvature of N--connecti\-on) is
\begin{eqnarray}
\ ^{^{\ast }\mathbf{J}}\mathbf{\Omega (X,Y)} &=& \mathbf{-[X,Y]+[\ ^{\ast
}JX,\ ^{\ast }JY]-\ ^{\ast }J[\ ^{\ast }JX,Y]-\ ^{\ast }J[X,\ ^{\ast }JY],}
\notag \\
\ ^{\mathbf{J}}\mathbf{\Omega (X,Y)} &=& \mathbf{\
-[X,Y]+[JX,JY]-J[JX,Y]-J[X,JY],}  \label{neijt}
\end{eqnarray}%
for any d--vectors $\mathbf{X}$ and $\mathbf{Y.}$
\end{definition}

Hereafter, for simplicity and if one shall not result in ambiguities, we
shall present the N--adapted component formulas for geometric objects on $%
T^{\ast }M$ (those for $TM$ being similar), or inversely.

With respect to N--adapted bases, the components of the Neijenhuis tensor $\
^{^{\ast }\mathbf{\ J}}\mathbf{\Omega}$ involve the coefficients $^{\ast
}\Omega _{ija}: $
\begin{equation}
\ ^{\ast }\Omega _{ija}=\frac{\partial \ ^{\ast }N_{ia}}{\partial x^{j}}-%
\frac{\partial \ ^{\ast }N_{ja}}{\partial x^{i}}+\ ^{\ast }N_{ib}\frac{%
\partial \ ^{\ast }N_{ja}}{\partial p_{b}}-\ ^{\ast }N_{jb}\frac{\partial \
^{\ast }N_{ia}}{\partial p_{b}} .  \label{nccurv}
\end{equation}
They define the coefficients of the N--connection curvature. One gets a
complex structure i.e $\ ^{^{\ast }\mathbf{\ J}}\mathbf{\Omega}=0$
 under some
quite complicated conditions on $g^{ab}(x,p)$ and $N_{ia}$ which will be not
written here.

It should be noted here that the N--adapted (co--) bases (\ref{dder})--(\ref%
{cddif}) are nonholonomic with nontrivial anholonomy coefficients. For
instance,
\begin{equation}
\lbrack \mathbf{e}_{\alpha },\mathbf{e}_{\beta }]=\mathbf{e}_{\alpha }%
\mathbf{e}_{\beta }-\mathbf{e}_{\beta }\mathbf{e}_{\alpha }=W_{\alpha \beta
}^{\gamma }\mathbf{e}_{\gamma }  \label{anhrel}
\end{equation}%
with (antisymmetric) anholonomy coefficients $W_{ia}^{b}=\partial
_{a}N_{i}^{b}$ and $W_{ji}^{a}=\Omega _{ij}^{a}, $ with :
\begin{equation*}
\Omega^a _{ij}=\frac{\partial N^a_{i}}{\partial x^{j}}-
\frac{\partial N^a_{j}}{\partial x^{i}}+N^b_{i}\frac{
\partial N^a_{j}}{\partial p_{b}}-N^b_{j}\frac{\partial
N^a_{i}}{\partial p_{b}} .
\end{equation*}

\subsection{Almost K\"{a}hler Lagrange--Hamilton structures}

We can adapt to N--connections various geometric structures on
$TM$ and $ T^{\ast }M.$ For instance, we can consider:

\begin{definition}
One calls an almost symplectic structure on $T^{\ast }M$ a nondegenerate
N--adapted 2--form
\begin{equation*}
\ ^{\intercal }\mathbf{\theta }=\frac{1}{2}\ ^{\intercal }\mathbf{\theta }%
_{\alpha \beta }(u)\ ^{\ast }\mathbf{e}^{\alpha }\wedge \ ^{\ast }\mathbf{e}%
^{\beta }.
\end{equation*}
\end{definition}

The following Proposition holds:

\begin{proposition}
\label{pr01}For any $\ ^{\intercal }\theta $ on $T^{\ast }M,$ there is a
unique N--connection $\ ^{\ast }\mathbf{N}=\{\ ^{\ast }N_{ia}\}$ satisfying
the conditions:%
\begin{equation}
\ ^{\intercal }\theta =(h\ ^{\ast }\mathbf{X},v\ ^{\ast }\mathbf{Y})=0%
\mbox{
and }\ ^{\intercal }\theta \doteq h\ ^{\intercal }\theta +v\ ^{\intercal }\theta ,
\label{aux02}
\end{equation}%
for any $\ ^{\ast }\mathbf{X}=h\ ^{\ast }\mathbf{X}+v\ ^{\ast }\mathbf{X,}$ $%
\ ^{\ast }\mathbf{Y}=h\ ^{\ast }\mathbf{Y}+v\ ^{\ast }\mathbf{Y}$ and $h\
^{\intercal }\theta (\ ^{\ast }\mathbf{X,\ ^{\ast }Y})\doteqdot \
^{\intercal }\theta (h\ ^{\ast }\mathbf{X,}h\ ^{\ast }\mathbf{Y}),$ $v\
^{\intercal }\theta (\ ^{\ast }\mathbf{X,\ ^{\ast }Y})\doteqdot \
^{\intercal }\theta (v\ ^{\ast }\mathbf{X,}v\ ^{\ast }\mathbf{Y}).$
\end{proposition}

\begin{proof}
For $\ ^{\ast }\mathbf{X=\ ^{\ast }e}_{\alpha }=(\ ^{\ast }\mathbf{e}_{i},\
^{\ast }e^{a})$ and $\ ^{\ast }\mathbf{Y=\ ^{\ast }e}_{\beta }=(\ ^{\ast }%
\mathbf{e}_{l},\ ^{\ast }e^{b}),$ where $\ ^{\ast }\mathbf{e}_{\alpha }$ is
a N--adapted basis\ of type (\ref{cdder}), we write the first equation in (%
\ref{aux02}) in the form%
\begin{equation*}
\mathbf{\ ^{\intercal }}\theta =\mathbf{\ ^{\intercal }}\theta (\mathbf{\
^{\ast }e}_{i},\mathbf{\ ^{\ast }}e^{a})=\mathbf{\ ^{\intercal }}\theta (%
\frac{\partial }{\partial x^{i}},\frac{\partial }{\partial p_{a}})-\mathbf{\
^{\ast }}N_{ib}\mathbf{\ ^{\intercal }}\theta (\frac{\partial }{\partial
p_{b}},\frac{\partial }{\partial p_{a}})=0.
\end{equation*}%
    Such conditions  uniquely define $\mathbf{\
^{\ast }}N_{i b}$  because $\mathbf{\ ^{\intercal }}\theta $ is
non--degenerate, that is $rank|\mathbf{\ ^{\intercal }}\theta (\frac{%
\partial }{\partial p_{b}},\frac{\partial }{\partial p_{a}})|=n.$
Setting locally
\begin{equation}
\mathbf{\ ^{\intercal }}\theta =\frac{1}{2}\mathbf{\ ^{\intercal }}\theta
_{ij}(u)e^{i}\wedge e^{j}+\frac{1}{2}\mathbf{\ ^{\intercal }}\theta ^{ab}(u)%
\mathbf{\ ^{\intercal }e}_{a}\wedge \mathbf{\ ^{\ast }e}_{b},  \label{aux03}
\end{equation}%
where the first term is for $h\mathbf{\ ^{\intercal }}\theta $ and the
second term is $v\mathbf{\ ^{\intercal }}\theta ,$ we get the second formula
in (\ref{aux02}). Finally, we note that in this proposition the constructed
N--connection $\mathbf{\ ^{\ast }}N_{ib},$ in general, is not a canonical
one (\ref{chnc}). $\square $
\end{proof}

\vskip3pt

In a similar form, as in Proposition \ref{pr01}, we can construct
a unique N--connection $\ \mathbf{N}=\{\ N_{i}^{a}\}$ for any
almost symplectic structure $\theta $ on $TM$ (from
formal point of view, we have to omit in formulas the symbols ''*'' and ''$%
\mathbf{\ ^{\intercal }}$''and use variables $y^{a}$ instead of $p_{a}).$

A N--connection $\ ^{\ast }\mathbf{N}$ (\ref{whitneyd}) defines a unique
decomposition of a d--vector $\ ^{\ast }\mathbf{X=\ ^{\ast }}X^{h}+\ ^{\ast
}X^{v}$ on $T^{\ast }M,$ for $\mathbf{\ ^{\ast }}X^{h}=h\ ^{\ast }\mathbf{X}$
and $\mathbf{\ ^{\ast }}X^{v}=v\ ^{\ast }\mathbf{X},$ where the projectors $%
h $ and $v$ defines respectively the distributions $^{\ast }\mathbf{N}$ and $%
^{\ast }\mathbf{V}$. They have the properties
\begin{equation*}
h+v=\mathbf{I},h^{2}=h,v^{2}=v,h\circ v= v\circ h=0.
\end{equation*}%
This allows us to introduce on $T^{\ast }M$ the almost product operator
\begin{equation*}
\ ^{\ast }\mathbf{P}\doteqdot I-2v=2h-I
\end{equation*}%
acting on $\ ^{\ast }\mathbf{e}_{\alpha }=(\ ^{\ast }\mathbf{e}_{i},\ ^{\ast
}e^{b})$ (\ref{cdder}) following formulas
\begin{equation*}
\ ^{\ast }\mathbf{P}(\ ^{\ast }\mathbf{e}_{i})=\ ^{\ast }\mathbf{e}_{i}%
\mbox{\ and \ }\ ^{\ast }\mathbf{P}(\ ^{\ast }e^{b})=-\ ^{\ast
}e^{b}.
\end{equation*}
In a similar form, a N--connection $\ \mathbf{N}$ (\ref{whitney})
induces an almost product structure $\mathbf{P}$
on $TM.$ One uses also the almost tangent operators $\ $%
\begin{eqnarray*}
\mathbb{J(}\mathbf{e}_{i}\mathbb{)} &=&e_{n+i}\mbox{\ and
\ }\ \mathbb{J}\left( e_{a}\right) =0,\mbox{ \ or \ }\mathbb{J=}\frac{%
\partial }{\partial y^{i}}\otimes dx^{i}; \\
\ ^{\ast }\mathbb{J(}\ ^{\ast }\mathbf{e}_{i}\mathbb{)} &=&\ ^{\ast }g_{ib}\
^{\ast }e^{b}\mbox{\ and
\ }\ ^{\ast }\mathbb{J}\left( \ ^{\ast }e^{b}\right) =0,\mbox{ \ or \ }%
\mathbb{J=}\ ^{\ast }g_{ia}\frac{\partial }{\partial p_{a}}\otimes dx^{i}.
\end{eqnarray*}%
The operators $\ ^{\ast }\mathbf{P,}\ ^{\ast }\mathbf{J}$ and $\ ^{\ast }%
\mathbb{J}$ are respectively $\mathcal{L}$--dual to $\ \mathbf{P,}\ \mathbf{J%
}$ and $\ \mathbb{J}$ if and only if $\ ^{\ast }\mathbf{N}$ and $\ \mathbf{N}
$ are $\mathcal{L}$--dual.

For the above--introduced almost complex and almost product operators, it is
straightforward to prove

\begin{proposition}
\label{propctf}Let $\left( \mathbf{N,}\ ^{\ast }\mathbf{N}\right) $ be a
pair of $\mathcal{L}$--dual N--connections. Then, we can construct canonical
d--tensor fields (defined respectively by $L(x,y)$ and $H(x,p)$ related by
Legendre transforms (\ref{legt1}) and/or (\ref{legt2})$):$%
\begin{equation*}
\mathbf{J=}-\delta _{i}^{a}e_{a}\otimes e^{i}+\delta _{a}^{i}\mathbf{e}%
_{i}\otimes \mathbf{e}^{a},\ ^{\ast }\mathbf{J=-\ ^{\ast }}g_{ia}\mathbf{\
^{\ast }}e^{a}\otimes \mathbf{\ ^{\ast }}e^{i}+\mathbf{\ ^{\ast }}g^{ia}%
\mathbf{\ ^{\ast }e}_{i}\otimes \mathbf{\ ^{\ast }e}_{a}
\end{equation*}%
corresponding to the $\mathcal{L}$--dual pair of almost complex structures $%
\left( \mathbf{J,}\ ^{\ast }\mathbf{J}\right) ,$
\begin{equation*}
\mathbf{P=e}_{i}\otimes e^{i}-e_{a}\otimes \mathbf{e}^{a},\ ^{\ast }\mathbf{%
P=\ ^{\ast }e}_{i}\otimes \mathbf{\ ^{\ast }}e^{i}-\mathbf{\ ^{\ast }}%
e^{a}\otimes \mathbf{\ ^{\ast }e}_{a}
\end{equation*}%
corresponding to the $\mathcal{L}$--dual pair of almost product structures $%
\left( \mathbf{P,}\ ^{\ast }\mathbf{P}\right) ,$ and almost symplectic
structures
\begin{equation}
\theta =g_{aj}(x,y)\mathbf{e}^{a}\wedge e^{i}\mbox{ and }\mathbf{\ ^{\ast }}%
\theta =\delta _{i}^{a}\mathbf{\ ^{\ast }e}_{a}\wedge \mathbf{\ ^{\ast }}%
e^{i}  \label{sympf}
\end{equation}
\end{proposition}

Let us consider an important example:

 A Finsler manifold (space)
$F^{n}=(M,F(x,y))$ is a particular case of Lagrange space, when
the regular Lagrangian $L=F^{2}$ is defined by a fundamental
Finsler function $F(x,y)$ satisfying the conditions: 1. the
positive function $F$ is differentiable function on
$\widetilde{TM}$
continuous on the null section of projection $\pi :TM\rightarrow M,$ 2. $%
F(x,\lambda y)=|\lambda |F(x,y),$ i.e. it is 1-homogeneous on the fibres of $%
TM,$ and 3. the Hessian (\ref{lm}) defined in this case by $F^{2}$ is
positively defined on $\widetilde{TM}.$ 

It is used also the notion of Cartan space $C^{n}=(M,C(x,p))$ for $%
H=C^{2}(x,p)$ as a particular (1-homogeneous on fiber coordinates) case of
Hamilton space when $C$ satisfies the same conditions as $F$ but with
respect to coordinates $p_{a}$ (in brief, we can say that Cartan spaces are
Finsler spaces on $T^{\ast }M,$ see details in \cite{mhss}). In a similar
manner as for Lagrange and Hamilton spaces, we can introduce the concept of $%
\mathcal{L}$--dual geometric objects on Finsler and Cartan spaces. For
simplicity, in this work we shall emphasize the bulk constructions for
Hamilton spaces considering that by Legendre transform we can generate
similar ones for Lagrange spaces and, in particular, for respective Finsler
and Cartan geometries.

\begin{definition}
\label{defaks}An almost Hermitian model of a cotangent bundle \ $T^{\ast }M$
(or tangent bundle $TM)$ equipped with a N--connection structure $\mathbf{\
^{\ast }N}$ (or $\mathbf{N})$ is defined by a triple $\mathbf{\ ^{\ast }H}%
^{2n}=(T^{\ast }M,\mathbf{\ ^{\ast }}\theta ,\mathbf{\ ^{\ast }J}),$ where $%
\mathbf{\ ^{\ast }\theta (\ ^{\ast }X,\ ^{\ast }Y)}\doteqdot \mathbf{\
^{\ast }g}\left( \mathbf{\ ^{\ast }JX,Y}\right) $ (or $\mathbf{H}%
^{2n}=(TM,\theta ,\mathbf{J}),$ where $\mathbf{\theta (X,Y)}\doteqdot
\mathbf{g}\left( \mathbf{JX,Y}\right) ).$ A space $\mathbf{\ ^{\ast }H}^{2n}$
is almost K\"{a}hler, denoted $\mathbf{\ ^{\ast }K}^{2n}$ if $d\mathbf{\
^{\ast }\theta }=0.$
\end{definition}

The following theorem holds:

\begin{theorem}
The Lagrange and Hamilton spaces can be represented as almost K\"{a}hler
spaces on, respectively, on $TM$ and $T^{\ast }M$ endowed with canonical
N--connection structures $\mathbf{N}$ (\ref{clnc}) and $\ ^{\ast }\mathbf{N}$
(\ref{chnc}).
\end{theorem}

\begin{proof}
It follows from the existence on $TM$ and $T^{\ast }M$ of canonical
1--forms, respectively, defined by a regular Lagrangian $L$ and Hamiltonian $%
H$ related by a Legendre transform,
\begin{equation*}
\omega =\frac{\partial L}{\partial y^{i}}e^{i}\mbox{ and }\ ^{\ast }\omega
=p_{i}dx^{i},
\end{equation*}%
for which
\begin{equation*}
\theta =d\omega \mbox{ and }\ \mathbf{^{\ast }\theta =}d\ ^{\ast }\omega ,
\end{equation*}%
see (\ref{csymps}). As a result, we get that $d\mathbf{\ \theta }=0$ and $d%
\mathbf{\ ^{\ast }\theta }=0,$ which correspond to the Definition \ref%
{defaks}. $\Box $
\end{proof}

\vskip5pt

In this paper, we shall work with almost K\"{a}hler models on (co) tangent
bundles defined canonically by (pseudo) Riemannian metrics on base
manifolds, see (\ref{eq01}), and/or (effective, or for regular mechanics)
Lagrangians (Hamiltonians). Finally, we emphasize that realistic classical
and quantum models are elaborated in explicit form for some classes of
linear connections defined to satisfy certain physical principles and
constructed geometrically to be adapted, or not, to a N--connection
structure. We shall perform such classical and quantum constructions in the
following sections.

\section{Nonlinear Connections and Almost Symplectic Geometry}

In this section, we consider the almost symplectic geometry induced by
regular Hamiltonians and corresponding canonical N--connections defined
naturally, for gravitational and/or geometric mechanics models, on (co)
tangent bundles.

\subsection{Canonical N--connections and d--connections for Lagran\-ge and
Hamilton spaces}

Let $D$ be a linear connection on $TM$ when for a
$\mathcal{L}$--duality between the tangent and corresponding
cotangent bundles there are defined pull--back and push--forward
maps. We can define a linear connection $\ ^{\ast }D$ on $T^{\ast
}M$ as follows:
\begin{equation*}
^{\ast }D_{\ ^{\ast }X}\ ^{\ast }Y\doteqdot (D_{\ ^{\circ }X}\
^{\circ }Y)^{\ast },
\end{equation*}
for any vector fields $\ ^{\ast }X$ and $\ ^{\ast }Y$ on $T^{\ast
}M.$ Inversely, for any linear connection $ ^{\ast }D$ on $T^{\ast
}M,$ we get a linear connection $^{\circ } D$ on $TM,$ following
the rule
\begin{equation*}
 ^{\circ }D_{X}Y\doteqdot (^{\ast }D_{ ^{\ast }X}\ ^{\ast
}Y)^{\circ },
\end{equation*}
for any vector fields $X$ and $Y$ on $TM.$

\begin{definition}
A linear connection $\mathbf{D}$  or ($\ ^{\ast }\mathbf{D})$ on
$TM$ (or $ T^{\ast }M)$ is a distinguished connection
(d--connection) if it is compatible with the almost product
structure $\mathbf{DP}=0$ (or $\ ^{\ast } \mathbf{D\ ^{\ast
}P}=0).$
\end{definition}

For $\mathcal{L}$--dual Lagrange and Hamilton spaces, one follows that $%
\mathbf{DP}=0$ induces $\ ^{\ast }\mathbf{D\ ^{\ast }P}=0,$ and
inversely. The coefficients of d--connections can be defined with
respect to N--adapted frames,
\begin{equation*}
\mathbf{D}_{\mathbf{e}_{\beta }}\mathbf{e}_{\gamma }\doteqdot \mathbf{\Gamma
}_{\ \beta \gamma }^{\alpha }\mathbf{e}_{\alpha }\mbox{ and }\mathbf{D}_{\
^{\ast }\mathbf{e}_{\beta }}\ ^{\ast }\mathbf{e}_{\gamma }\doteqdot \ ^{\ast
}\mathbf{\Gamma }_{\ \beta \gamma }^{\alpha }\ ^{\ast }\mathbf{e}_{\alpha }
\end{equation*}%
with corresponding N--adapted splitting,
\begin{equation*}
\mathbf{D}_{\mathbf{e}_{k}}\mathbf{e}_{j}\doteqdot L_{\ jk}^{i}\mathbf{e}%
_{i},\mathbf{D}_{\mathbf{e}_{k}}e_{b}\doteqdot \acute{L}_{\ bk}^{a}e_{a},%
\mathbf{D}_{e_{c}}\mathbf{e}_{j}\doteqdot \acute{C}_{\ jc}^{i}\mathbf{e}_{i},%
\mathbf{D}_{e_{c}}e_{b}\doteqdot C_{\ bc}^{a}e_{a}
\end{equation*}%
and%
\begin{eqnarray*}
\ ^{\ast }\mathbf{D}_{\ ^{\ast }\mathbf{e}_{k}}\ ^{\ast }\mathbf{e}_{j}
&\doteqdot &\ ^{\ast }L_{\ jk}^{i}\ ^{\ast }\mathbf{e}_{i},\ ^{\ast }\mathbf{%
D}_{\mathbf{e}_{k}}\ ^{\ast }e^{b}\doteqdot -\ ^{\ast }\acute{L}_{a\ k}^{\
b}\ ^{\ast }e^{a}, \\
\ ^{\ast }\mathbf{D}_{\ ^{\ast }e^{c}}\ ^{\ast }\mathbf{e}_{j} &\doteqdot &\
^{\ast }\acute{C}_{\ j}^{i\ c}\ ^{\ast }\mathbf{e}_{i},\ ^{\ast }\mathbf{D}%
_{\ ^{\ast }e^{c}}\ ^{\ast }e^{b}\doteqdot -\ ^{\ast }C_{a}^{\ bc}\ ^{\ast
}e^{a},
\end{eqnarray*}%
when
\begin{equation*}
\mathbf{\Gamma }_{\ \beta \gamma }^{\alpha }=\{L_{\ jk}^{i},\acute{L}_{\
bk}^{a},\acute{C}_{\ jc}^{i},C_{\ bc}^{a}\}\mbox{ and }\ ^{\ast }\mathbf{%
\Gamma }_{\ \beta \gamma }^{\alpha }=\{\ ^{\ast }L_{\ jk}^{i},\ ^{\ast }%
\acute{L}_{a\ k}^{\ b},\ ^{\ast }\acute{C}_{\ j}^{i\ c},\ ^{\ast }C_{a}^{\
bc}\}
\end{equation*}%
define corresponding h-- and v--splitting of covariant derivatives
\begin{equation*}
\mathbf{D=}\left( \mathbf{\ }_{h}\mathbf{D,\ }_{v}\mathbf{D}\right)
\mbox{
and }\ ^{\ast }\mathbf{D=}\left( \mathbf{\ }_{h}^{\ast }\mathbf{D,\ }%
_{v}^{\ast }\mathbf{D}\right) ,
\end{equation*}%
where $\ _{h}\mathbf{D}=\{L_{\ jk}^{i},\acute{L}_{\ bk}^{a}\},\ _{v}\mathbf{D%
}=\{\acute{C}_{\ jc}^{i},C_{\ bc}^{a}\}$ and $\ _{h}^{\ast }\mathbf{D}=\{\
^{\ast }L_{\ jk}^{i},\ ^{\ast }\acute{L}_{a\ k}^{\ b}\},$ $\ _{v}^{\ast }%
\mathbf{D}=\{\ ^{\ast }\acute{C}_{\ j}^{i\ c},\ ^{\ast }C_{a}^{\ bc}\}.$

We shall work with a more special class of d--connections:

\begin{definition}
\label{defnlinc}A linear connection $\ ^{n}\mathbf{D}$ (or $\ ^{\ast n}%
\mathbf{D)}$ on $TM$ (or $T^{\ast }M)$ is N--linear if it preserves under
parallelism the Whitney sum $\mathbf{N}$ (\ref{whitney}) (or $\ ^{\ast }%
\mathbf{N}$ (\ref{whitneyd})) i.e it is a $d$-connection and is compatible
with the almost tangent structure $\mathbb{J}$ (or $\ ^{\ast }\mathbb{J)}$
i.e. $\ ^{n}\mathbf{D}\mathbb{J}=0$ (or $\ ^{\ast n}\mathbf{D}\ ^{\ast }%
\mathbb{J}=0).$
\end{definition}

This is a class of N--linear connections characterized by the property:

\begin{proposition}
The N--adapted coefficients of N--linear connections\newline
$\ ^{n}\mathbf{D}\doteqdot \{\ ^{n}\Gamma _{\ \beta \gamma }^{\alpha }=(L_{\
jk}^{i},\acute{C}_{\ jc}^{i})\}$ and $\ ^{\ast n}\mathbf{D}\doteqdot \{\
^{n\ast }\Gamma _{\ \beta \gamma }^{\alpha }=(\ ^{\ast }L_{\ jk}^{i},\
^{\ast }\acute{C}_{\ j}^{i\ c})\}$ of respective $\mathcal{L}$--dual
Lagrange and Hamilton spaces are:
\begin{equation}
\ ^{n}\mathbf{\Gamma }_{\ \beta \gamma }^{\alpha }=\{L_{\ jk}^{i},\acute{L}%
_{\ n+jk}^{n+i}=L_{\ jk}^{i},\acute{C}_{\ jc}^{i},C_{\ n+jc}^{n+i}=\acute{C}%
_{\ jc}^{i}\}  \label{nlinc}
\end{equation}%
and
\begin{equation}
\ ^{\ast n}\mathbf{\Gamma }_{\ \beta \gamma }^{\alpha }=\{\ ^{\ast }L_{\
jk}^{i},\ ^{\ast }\acute{L}_{n+j\ k}^{\ n+i}=\ ^{\ast }L_{\ jk}^{i},\ ^{\ast
}\acute{C}_{\ j}^{i\ c},\ ^{\ast }C_{n+j}^{\ n+i\ c}=\ ^{\ast }\acute{C}_{\
j}^{i\ c}\}.  \label{nlincd}
\end{equation}
\end{proposition}

\begin{proof}
By a straightforward computation for coefficients (\ref{nlinc}) and (\ref%
{nlincd}), when, for instance, $\ ^{\ast n}\mathbf{D}_{\alpha }=(\
_{h}^{\ast }D_{i},\ _{v}^{\ast }D^{a}),$ for $\ _{h}^{\ast }D=\{\ ^{\ast
}L_{\ jk}^{i}\}$ and $\ _{v}^{\ast }D=\{\ ^{\ast }\acute{C}_{\ j}^{i\ c}\},$
we can verify that the conditions considered in Definition \ref{defnlinc}
are satisfied. $\Box $
\end{proof}

\vskip5pt

The connection 1--form of N--linear connection $\ ^{\ast n}\mathbf{\Gamma }%
_{\ \beta }^{\alpha }\mathbf{=}\ ^{\ast n}\mathbf{\Gamma }_{\ \beta \gamma
}^{\alpha }\ ^{\ast }\mathbf{e}^{\gamma } = $ $\{\ ^{\ast n}\mathbf{\Gamma }%
_{\ j}^{i}\}$ is defined by
\begin{equation}
\ ^{\ast n}\mathbf{\Gamma }_{\ j}^{i}=\ ^{\ast }L_{\ jk}^{i}\ ^{\ast
}e^{k}+\ ^{\ast }\acute{C}_{\ j}^{i\ c}\ ^{\ast }\mathbf{e}_{c},
\label{nlc1f}
\end{equation}%
where the v--components $\ ^{\ast n}\mathbf{\Gamma }_{\ b}^{a}$ are
identified with the h--components $\ ^{\ast n}\mathbf{\Gamma }_{\ j}^{i}$
following formulas $\ ^{\ast n}\mathbf{\Gamma }_{\ n+j}^{n+i}=\ ^{\ast n}%
\mathbf{\Gamma }_{\ b}^{a}.$ The following theorem holds:

\begin{theorem}
\label{tcsteq}On a Hamilton space, the structure equations for $\ ^{\ast n}%
\mathbf{\Gamma }_{\ j}^{i}$ are
\begin{eqnarray*}
d\ ^{\ast }e^{k}-\ ^{\ast }e^{j}\wedge \ ^{\ast n}\mathbf{\Gamma }_{\ j}^{k}
&=&-\ ^{\ast }\mathcal{T}^{k}, \\
d\ ^{\ast }\mathbf{e}_{c}+\ ^{\ast }\mathbf{e}_{a}\wedge \ ^{\ast n}\mathbf{%
\Gamma }_{\ c}^{a} &=&-\ ^{\ast }\mathcal{T}_{c}, \\
d\ ^{\ast n}\mathbf{\Gamma }_{\ j}^{i}-\ ^{\ast n}\mathbf{\Gamma }_{\
k}^{i}\wedge \ ^{\ast n}\mathbf{\Gamma }_{\ j}^{k} &=&-\ ^{\ast }\mathcal{R}%
_{\ j}^{i},
\end{eqnarray*}%
where the 2--form of torsion $\mathcal{T}^{\alpha }=\{\mathcal{T}^{k},%
\mathcal{T}_{c}\}$ is computed
\begin{eqnarray*}
\ ^{\ast }\mathcal{T}^{k} &=&\frac{1}{2}\ ^{\ast }T_{\ ij}^{k}\ ^{\ast
}e^{i}\wedge \ ^{\ast }e^{j}+\ ^{\ast }\acute{C}_{\ j}^{k\ c}\ ^{\ast
}e^{j}\wedge \ ^{\ast }\mathbf{e}_{c},\  \\
\ ^{\ast }\mathcal{T}_{a} &=&\frac{1}{2}\ ^{\ast }\Omega _{ija}\ ^{\ast
}e^{i}\wedge \ ^{\ast }e^{j}+\frac{1}{2}\ ^{\ast }P_{aic}\ ^{\ast
}e^{i}\wedge \ ^{\ast }\mathbf{e}_{c}+\frac{1}{2}\ ^{\ast }S_{a}^{\ bc}\
^{\ast }\mathbf{e}_{b}\wedge \ ^{\ast }\mathbf{e}^{c},
\end{eqnarray*}%
for the N--connection curvature $\ ^{\ast }\Omega _{ija}$ (\ref{nccurv}) and
\begin{eqnarray}
\ ^{\ast }T_{\ ij}^{k} &=&\ ^{\ast }L_{\ ij}^{k}-\ ^{\ast }L_{\ ji}^{k},\
^{\ast }S_{a}^{\ bc}=\ ^{\ast }\acute{C}_{\ a}^{b\ c}-\ ^{\ast }\acute{C}_{\
a}^{c\ b},\   \notag \\
\ ^{\ast }P_{aic} &=&\ ^{\ast }g_{ae}\left( \ ^{\ast }L_{\ ic}^{e}-\ ^{\ast
}e^{e}(\ ^{\ast }N_{ic})\right) ,  \label{tcncs}
\end{eqnarray}
and the 2--form of curvature $\ ^{\ast }\mathcal{R}_{\ \beta }^{\alpha
}=\left( \ ^{\ast }\mathcal{R}_{\ j}^{i},\ ^{\ast }\mathcal{R}_{\
b}^{a}\right) ,$ with $\ ^{\ast }\mathcal{R}_{\ j}^{i}=\ ^{\ast }\mathcal{R}%
_{\ n+j}^{n+i},$ is computed%
\begin{equation*}
\ ^{\ast }\mathcal{R}_{\ j}^{i}=\frac{1}{2}\ ^{\ast }R_{\ jkm}^{i}\ ^{\ast
}e^{k}\wedge \ ^{\ast }e^{m}+\ ^{\ast }P_{\ jk}^{i\ \ c}\ ^{\ast
}e^{k}\wedge \ ^{\ast }\mathbf{e}_{c}+\frac{1}{2}\ ^{\ast }S_{\ j}^{i\ bc}\
^{\ast }\mathbf{e}_{b}\wedge \ ^{\ast }\mathbf{e}_{c},
\end{equation*}%
where
\begin{eqnarray}
\ ^{\ast }R_{\ jkm}^{i} &=&\ ^{\ast }e_{m}(\ ^{\ast }L_{\ jk}^{i})-\ ^{\ast
}e_{k}(\ ^{\ast }L_{\ jm}^{i})  \label{ccncs} \\
&&+\ ^{\ast }L_{\ jk}^{o}\ ^{\ast }L_{\ om}^{i}-\ ^{\ast }L_{\ jm}^{o}\
^{\ast }L_{\ ok}^{i}+\ ^{\ast }\acute{C}_{\ j}^{i\ o}\ ^{\ast }\Omega _{oka},
\notag \\
\ ^{\ast }P_{\ jk}^{i\ \ c} &=&\ ^{\ast }e^{c}(\ ^{\ast }L_{\ jk}^{i})-\
_{h}^{\ast }D_{k}(\ ^{\ast }\acute{C}_{\ j}^{i\ c})+\ ^{\ast }\acute{C}_{\
j}^{i\ o}\ ^{\ast }P_{\ ko}^{c},  \notag \\
\ ^{\ast }S_{\ j}^{i\ bc} &=&\ ^{\ast }e^{c}(\ ^{\ast }\acute{C}_{\ j}^{i\
b})-\ ^{\ast }e^{b}(\ ^{\ast }\acute{C}_{\ j}^{i\ c})+\ ^{\ast }\acute{C}_{\
j}^{k\ b}\ ^{\ast }\acute{C}_{\ k}^{i\ c}-\ ^{\ast }\acute{C}_{\ j}^{k\ c}\
^{\ast }\acute{C}_{\ k}^{i\ b},  \notag
\end{eqnarray}%
for ''non--boldface'' $\ ^{\ast }e^{c}=\partial /\partial p_{c}.$
\end{theorem}

\begin{proof}
It is a straightforward differential computation for 1--form
(\ref{nlc1f}). $ \Box $
\end{proof}

\vskip5pt

There is a ''$\mathcal{L}$--dual'' Theorem for Lagrange spaces \cite{ma},
see generalizations of nonholonomic manifolds and deformation quantization
of gravity in Ref. \cite{vrfg,vqgr4}, similarly to Theorem \ref{tcsteq}.
From formal point of view, we have to change $H$ into $L$ and consider the
constructions on $TM,$ omitting the labels ''*'' and using coordinates $%
(x^{i},y^{a})$ instead of $(x^{i},p_{a}).$

\subsection{Hamilton--Fedosov spaces and almost K\"{a}hler structures}

There are canonical N--linear connections on $TM$ and $T^{\ast }M$
completely defined, respectively, by the fundamental Lagrange (see \cite{ma}
and, for applications to geometric quantization, \cite{vqgr1,vqgr4}) and
Hamilton functions.

\begin{theorem}
\label{cnlinc}There exists a canonical N--linear connection\newline
$\ ^{\ast }\widehat{\mathbf{D}}\doteqdot \{\ ^{\ast }\widehat{\mathbf{\Gamma
}}_{\ \beta \gamma }^{\alpha }=(\ ^{\ast }\widehat{L}_{\ jk}^{i},\ ^{\ast }%
\widehat{\acute{C}}_{\ j}^{i\ c})\}$ on a Hamilton space $H^{n}=(M,H(x,p))$
endowed with canonical N--connection $\ ^{\ast }N_{ij}$ (\ref{chnc})
satisfying the conditions: 1) $\ _{h}^{\ast }\widehat{D}_{i}\ ^{\ast
}g^{kj}=0$ and $\ _{v}^{\ast }\widehat{D}_{a}\ ^{\ast }g^{kj}=0,$ 2) $\
^{\ast }T_{\ ij}^{k}=0$ and $\ ^{\ast }S_{a}^{\ bc}=0$ and 3) $\ ^{\ast }%
\widehat{\mathbf{D}}$ is completely defined by $H(x,p),$ i.e. by d--metric (%
\ref{dmctb}).
\end{theorem}

\begin{proof}
Let us consider the N--adapted coefficients
\begin{eqnarray}
\ ^{\ast }\widehat{L}_{\ jk}^{i} &=&\frac{1}{2}\ ^{\ast }g^{is}\left( \
^{\ast }e_{j}(\ ^{\ast }g_{sk})+\ ^{\ast }e_{k}(\ ^{\ast }g_{js})-\ ^{\ast
}e_{s}(\ ^{\ast }g_{jk})\right) ,  \notag \\
\ ^{\ast }\widehat{\acute{C}}_{\ j}^{i\ c} &=&-\frac{1}{2}\ ^{\ast }g_{js}\
^{\ast }e^{c}(\ ^{\ast }g^{si})  \label{clnchs}
\end{eqnarray}%
defined with respect to N--adapted bases $\ ^{\ast }\mathbf{e}_{\alpha }$ (%
\ref{cdder}) and coframes $\ ^{\ast }\mathbf{e}^{\alpha }$ (\ref{cddif})
defined by $\ ^{\ast }N_{ij}$ (\ref{chnc}). By a direct computation, we can
verify that the conditions of this Theorem are satisfied only for such
coefficients and their coordinate/frame transform. $\Box $
\end{proof}

\vskip5pt

The coefficients (\ref{clnchs}) are just the Christoffel symbols on (co)
tangent space $T^{\ast }M,$ defined by $H.$ Nevertheless, for a Hamilton
space endowed with canonical N--connection, d--connection and metric
structure, there are nontrivial torsion components induced by the
nonholonomic distribution defined by $H,$ see $\ ^{\ast }P_{aic}$ (\ref%
{tcncs}).

From Theorem \ref{cnlinc}, one follows an important property:

\begin{corollary}
The canonical d--connection $\ ^{\ast }\widehat{\mathbf{D}}$ is an almost
symplectic d--connection satisfying the conditions
\begin{equation}
\ ^{\ast }\widehat{\mathbf{D}}\mathbf{\ ^{\ast }\theta }=0\mbox{\  and\  }\
^{\ast }\widehat{\mathbf{D}}\mathbf{\ ^{\ast }J}=0  \label{asccond}
\end{equation}%
and being completely defined by a Hamiltonian $H(x,p)$ for $\mathbf{\ ^{\ast
}\theta (\ ^{\ast }X,\ .)}\doteqdot \mathbf{\ ^{\ast }g}\left( \mathbf{\
^{\ast }JX,\ .}\right) .$
\end{corollary}

\begin{proof}
A Hamilton space can be equivalently transformed into almost K\"{a}hler
space $\mathbf{\ ^{\ast }K}^{2n}$ , see Definition \ref{defaks}. Considering
the d--tensor fields associated to $\mathbf{\ ^{\ast }\theta }$ and $\mathbf{%
\ ^{\ast }J,}$ see Proposition \ref{propctf}, and a covariant N--adapted
calculus defined by the canonical d--connection coefficients (\ref{clnchs})
we can verify that the almost symplectic structure compatibility conditions (%
\ref{asccond}) are satisfied. $\Box $
\end{proof}

\vskip5pt

In Ref. \cite{esv}, it was introduced the concept of Lagrange--Fedosov
manifold as nonholonomic manifold with the N--connection and almost
symplectic structure defined by a fundamental (in genera, effective)
Lagrange function $L(x,y).$ On cotangent bundles, we can consider

\begin{definition}
A Hamilton--Fedosov space is a cotangent bundle endowed with
canonical N--connection and the almost K\"{a}hler structure
induced by a fundamental Hamilton function $H(x,p).$
\end{definition}

There are Hamilton--Fedosov spaces defined completely by a lift (\ref{eq01})
of a (pseudo) Riemannian metric on base $M.$

\begin{theorem}
Any Einstein manifold associated to a solution of (\ref{einstm}), for a lift
(\ref{eq01}) on $T^{\ast }M,$ defines canonically a Hamilton--Fedosov space.
\end{theorem}

\begin{proof}
We sketch the idea for such constructions. Let us fix any values $e_{\ i^{\prime
}}^{i}(x,p)$ in (\ref{eq01}) and associate $g^{ab}(x,p)$ to a $\ ^{\ast
}g^{ab}(x,p)$ (\ref{hm}). This define correspondingly the values $\ ^{\ast }%
\mathbf{g}$ (\ref{dmctb}) and $\ ^{\ast }\mathbf{N}$ (\ref{chnc}).
As a result, we construct an effective Hamilton space, which can
be modelled as a canonical almost symplectic structure as we
described above. For classical configurations, the values $e_{\
i^{\prime }}^{i}$ can be $\delta _{i^{\prime }}^{i},$ but for
quantum models they should defined by a scheme of
de--quantization, or semi--classical approximation in quantum
gravity.

There is a particular case of  Cartan--Fedosov spaces with
0-homogeneous on variables "p" components $e^i_{i'}(x,p) $
resulting in a similar homogeneity for $g^{ab} (x,p)$ and $H =
g^{ab}p_ap_b$ when $^*N$ is determined from (15) with $ ^*g$  from (21).

The problem is more sophisticate in the case of  general Hamilton--Fedosov spaces. For certain physical important four dimensional spaces (used in general relativity) with nonholonomic splitting of dimensions as 2+2 and, for instance, if $dim M=2,$ we can fix such local coordinate systems when (3) is integrable for  certain solvable partial differential equations for $e^i_{i'}(x,p),$ but this may not hold true for other parameterizations and higher dimensions. A general  approach should include the case of Eisenhart--Hamilton spaces and their Fedosov quantum
deformation analogs, with both symmetric and nonsymmetric
components for (3) resulting because of any general quantum
nonholonomic Legendre transform. Such constructions should dub on cotangent bundles those for Eisenhart-Lagrange/- Finsler
spaces (see Chapter 8 in [33] and Refs. [52,53]). The length of this paper does not allow us to present a detailed proof because
it is connected with a sophisticate geometric techniques for
nonlinear connections and nonsymmetric metrics arising both in the
case of nonholonomic Ricci flows and quantum nonholonomic
deformations and/or symplectic transform in gravity and geometric
mechanics, see recent results in Refs. [54-56].
 $
\Box $
\end{proof}

\vskip5pt

Finally, we note that similar almost symplectic models can be performed for
Cartan spaces when the Hamiltonian is homogeneous on vertical coordinates.

\section{Fedosov Operator--Pairs for Hamilton Spaces}

In this section, we shall apply the method of deformation
quantization elaborated in Refs. \cite{fed1,fed2,karabeg1} to
define two classes of canonical operators which are necessary to
quantize the Hamilton--Fedosov spaces and related subspaces on
cotangent bundles defined by lifts of Einstein metrics. We shall
address precisely the question how the geometry of cotangent
bundles and related deformation quantization change under
symplectic transforms and elaborate a formalism which preserves
the form of Hamilton--Jacobi equations both on classical and
quantum level.

\subsection{Canonical Fedosov--Hamilton operators}

The formalism of deformation quantization can be developed by
using the space $C^{\infty }(\mathbf{\ ^{\ast }K}^{2n})[[v]]$ of
formal series in the variable $v$ with coefficients from
$C^{\infty }(\mathbf{\ ^{\ast }K}^{2n})$ on a almost Poisson
manifold $(\mathbf{\ ^{\ast }K}^{2n},\{\cdot ,\cdot \}),$ see the
almost symplectic form $ \mathbf{ ^{\ast }}\theta $ and the
Poisson brackets (\ref{poisbr}).  Using the associative algebra
structure on $C^{\infty }( \mathbf{\ ^{\ast }K}^{2n})[[v]]$ with a
$v$--linear and $v$--adically continuous star product
\begin{equation}
\ ^{1}f\ast \ ^{2}f=\sum\limits_{r=0}^{\infty }\ _{r}C(\ ^{1}f,\ ^{2}f)\
v^{r},  \label{starp}
\end{equation}%
where $\ _{r}C,r\geq 0,$ are bilinear operators on $C^{\infty
}(\mathbf{\ ^{\ast }K}^{2n})$ with $\ _{0}C(\ ^{1}f,\ ^{2}f)=\
^{1}f\ ^{2}f$ and $\ _{1}C(\ ^{1}f,\ ^{2}f)-\ _{1}C(\ ^{2}f,\
^{1}f)=i\{\ ^{1}f,\ ^{2}f\};$\ $\ $for $i$ being the complex
unity, we can construct a formal Wick product
\begin{equation}
\ ^{1}a\circ \ ^{2}a\ (z)\doteqdot \exp \left( i\frac{v}{2}\ \mathbf{^{\ast
}\Lambda }^{\alpha \beta }\frac{\partial ^{2}}{\partial z^{\alpha }\partial
z_{[1]}^{\beta }}\right) \ ^{1}a(z)\ ^{2}a(z_{[1]})\mid _{z=z_{[1]}},
\label{fpr}
\end{equation}%
for two elements $a$ and $b$ defined by series of the type
\begin{equation}
a(v,z)=\sum\limits_{r\geq 0,|\{\alpha \}|\geq 0}\ a_{r,\{\alpha
\}}(u)z^{\{\alpha \}}\ v^{r},  \label{formser}
\end{equation}%
where by $\{\alpha \}$ we label a multi--index and $\ \mathbf{^{\ast
}\Lambda }^{\alpha \beta }\doteqdot \ \mathbf{^{\ast }}\theta ^{\alpha \beta
}-i\ \mathbf{^{\ast }g}^{\alpha \beta },$ where $\ \mathbf{^{\ast }}\theta
^{\alpha \beta }$ is the symplectic form (\ref{sympf}), with ''up'' indices,
and $\ \mathbf{^{\ast }g}^{\alpha \beta }$ is the inverse to d--tensor (\ref%
{dmctb}). This defines a formal Wick algebra $\ \mathbf{^{\ast }W}_{u}$
associated with the tangent space $T_{u}\mathbf{\ ^{\ast }K}^{2n},$ for $%
u\in \mathbf{\ ^{\ast }K}^{2n},$ where the local coordinates on $\mathbf{\
^{\ast }K}^{2n}$ are parameterized in the form $u=\{u^{\alpha }\}$ and the
local coordinates on $T_{u}\mathbf{\ ^{\ast }K}^{2n}$ are labelled $%
(u,z)=(u^{\alpha },z^{\beta }),$ where $z^{\beta }$ are fiber coordinates.

We trivially extend the fibre product (\ref{fpr}) to the space of $\mathbf{\
^{\ast }W}$--valued N--adapted differential forms $\mathbf{\ ^{\ast }}%
\mathcal{W}\otimes \Lambda ,$ where by $\Lambda $ we note the usual exterior
product of the scalar forms and $\mathbf{\ ^{\ast }}\mathcal{W}$ is the
sheaf of smooth sections of $\mathbf{\ ^{\ast }W.}$ There is a standard
grading on $\Lambda $ noted $\deg _{a}.$ We also introduce gradings $\deg
_{v},\deg _{s},\deg _{a}$ on$\mathbf{\ ^{\ast }}\mathcal{W}\otimes \Lambda $
defined on homogeneous elements $v,z^{\alpha },\mathbf{\ ^{\ast }e}^{\alpha
},$ when $\deg _{v}(v)=1,$ $\deg _{s}(z^{\alpha })=1,$ $\deg _{a}(\mathbf{\
^{\ast }e}^{\alpha })=1,$ and all other gradings of the elements $%
v,z^{\alpha },\mathbf{\ ^{\ast }e}^{\alpha }$ are set to zero. As a result,
the product $\circ $ from (\ref{fpr}) on $\ \mathbf{^{\ast }}\mathcal{W}%
\otimes \mathbf{\Lambda }$ is bi-graded, written as w.r.t the grading $%
Deg=2\deg _{v}+\deg _{s}$ and the grading $\deg _{a}.$

The canonical d--connection $\ ^{\ast }\widehat{\mathbf{D}}\doteqdot \{\
^{\ast }\widehat{\mathbf{\Gamma }}_{\ \beta \gamma }^{\alpha }=(\ ^{\ast }%
\widehat{L}_{\ jk}^{i},\ ^{\ast }\widehat{\acute{C}}_{\ j}^{i\ c})\}$ with
coefficients (\ref{clnchs}) can be extended to an operator
\begin{equation}
\ ^{\ast }\widehat{\mathbf{D}}\left( a\otimes \lambda \right) \doteqdot
\left( \ ^{\ast }\mathbf{e}_{\alpha }(a)-u^{\beta }\ \ ^{\ast }\widehat{%
\mathbf{\Gamma }}\mathbf{_{\alpha \beta }^{\gamma }\ }^{z\ast }\mathbf{e}%
_{\alpha }(a)\right) \otimes (\ ^{\ast }\mathbf{e}^{\alpha }\wedge \lambda
)+a\otimes d\lambda ,  \label{cdcop}
\end{equation}%
on $\ \mathbf{^{\ast }}\mathcal{W}\otimes \Lambda ,$ where the N--adapted
basis $^{z\ast }\mathbf{e}_{\alpha }$ is $\ ^{\ast }\mathbf{e}_{\alpha }$
redefined in $z$--variables. This canonical almost symplectic d--connection $%
\ ^{\ast }\widehat{\mathbf{D}}$ is a N--adapted $\deg _{a}$--graded
derivation of the distinguished algebra $\left( \ \mathbf{^{\ast }}\mathcal{W%
}\otimes \mathbf{\Lambda ,\circ }\right) ,$ in brief, called d--algebra:
this follows from formulas (\ref{fpr}) and (\ref{cdcop})).

\begin{definition}
The Fedosov--Hamilton operators $\ ^{\ast }\mathbf{\delta }$ and $\ ^{\ast }%
\mathbf{\delta }^{-1}$ on$\ \ ^{\ast }\mathcal{W}\otimes \mathbf{\Lambda ,}$
are defined%
\begin{equation}
\ ^{\ast }\mathbf{\delta }(a)=\ ^{\ast }\mathbf{e}^{\alpha }\wedge \mathbf{\
}^{z\ast }\mathbf{e}_{\alpha }(a),\ \mbox{and\ }\ \ ^{\ast }\mathbf{\delta }%
^{-1}(a)=\left\{
\begin{array}{c}
\frac{i}{p+q}z^{\alpha }\ \ ^{\ast }\mathbf{e}_{\alpha }(a),\mbox{ if }p+q>0,
\\
{\qquad 0},\mbox{ if }p=q=0,%
\end{array}%
\right.  \label{feddop}
\end{equation}%
where any $a\in \ \mathbf{^{\ast }}\mathcal{W}\otimes \mathbf{\Lambda }$ is
homogeneous w.r.t. the grading $\deg _{s}$ and $\deg _{a}$ with $\deg
_{s}(a)=p$ and $\deg _{a}(a)=q.$
\end{definition}

The d--operators (\ref{feddop}) define the formula $a=(\ ^{\ast }\mathbf{%
\delta }\ \ ^{\ast }\mathbf{\delta }^{-1}+\ ^{\ast }\mathbf{\delta }^{-1}\ \
^{\ast }\mathbf{\delta }+\sigma )(a),$ where $a\longmapsto \sigma (a)$ is
the projection on the $(\deg _{s},\deg _{a})$--bihomogeneous part of $a$ of
degree zero, $\deg _{s}(a)=\deg _{a}(a)=0;$ $\ ^{\ast }\mathbf{\delta }$ is
also a $\deg _{a}$--graded derivation of the d--algebra $\left( \ \mathbf{%
^{\ast }}\mathcal{W}\otimes \mathbf{\Lambda ,\circ }\right) .$

\subsection{Fedosov--Hamilton operator--pairs}

Having defined d--operators (\ref{feddop}), we can perform a Fedosov type
quantization of Hamilton spaces in a $\mathcal{L}$--dual form to Lagrange
spaces \cite{vqgr1,vqgr4}. \ Nevertheless, such constructions would not
reflect completely the symplectic properties of Hamilton spaces and their
quantum deformations. For a cotangent bundle $\pi ^{\ast }:T^{\ast
}M\rightarrow M,$ a N--connection $\ ^{\ast }\mathbf{N}$ (\ref{chnc}) is a
supplementary distribution $hT^{\ast }M$ of the vertical distribution $%
vT^{\ast }M=\ker $ $\ ^{\intercal }\pi ^{\ast },$ where $\ ^{\intercal }\pi
^{\ast }$ is the tangent map of $\pi ^{\ast }.$ It is often more convenient
to consider a N--connection as an almost product structure $^{\ast }\mathbf{%
P,}$ see Proposition \ref{propctf}, such that $vT^{\ast }M=\ker(I+^{\ast }%
\mathbf{P}).$ If $f\in Diff(T^{\ast }M),$ the push--forward of $\ ^{\ast }%
\mathbf{N}$ by $f$ generally fails to be a connection. This constrains us to
extend the definition of connection, see details in Chapter 8 of \cite{mhss}.

\begin{definition}
A connection--pair $\ ^{\ast }\phi $ is an almost product structure $\
^{\ast }\phi $ on $T^{\ast }M$ such that the horizontal bundle $%
hT^{\ast }M\doteqdot \ker (I-\ ^{\ast }\phi )$ is supplementary to $vT^{\ast
}M$ and the oblique bundle $wT^{\ast }M\doteqdot \ker (I+\ ^{\ast }\phi )$
define the nonholonomic splitting (Whitney sum)
\begin{equation}
TT^{\ast }M=hT^{\ast }M\oplus wT^{\ast }M.  \label{cpspl}
\end{equation}
\end{definition}

We can consider preferred $\ ^{\ast }\phi $--adapted frame and coframe
structures induced by the coefficients of N--connection:

\begin{proposition}
There are a canonical connection--pair $\ ^{\ast }\phi $ and associated
frames and coframes:%
\begin{eqnarray}
\ ^{\phi }\mathbf{e}_{\alpha } &=&\left( \ ^{\phi }e_{i}=\frac{\partial }{%
\partial x^{i}}+\ ^{\ast }N_{ia}\frac{\partial }{\partial p_{a}},\ ^{w}e^{b}=%
\frac{\partial }{\partial p_{b}}-\ ^{\ast }g^{bi}\ ^{\phi }e_{i}\right)
\label{fder} \\
\ ^{\phi }\mathbf{e}^{\alpha } &=&\left( \ ^{\phi }e^{i}=dx^{i}+\ ^{\ast
}g^{bi}\ ^{w}e_{b},\ ^{w}e_{b}=dp_{b}-\ ^{\ast }N_{ib}dx^{i}\right) ,
\label{fdif}
\end{eqnarray}%
where $^{\ast }g^{ab}$ (\ref{hm}) and $^{\ast }N_{ia}$ (\ref{chnc}) are
generated by a Hamilton fundamental function $H(x,p).$
\end{proposition}

\begin{proof}
The proof is similar to that for Proposition \ref{pnaf} but adapted both the
the splitting (\ref{cpspl}), with respective h- and v--projections, $%
2h^{\prime }=I+\ ^{\ast }\phi $ and $2w=I-\ ^{\ast }\phi \mathbf{,}$ and to
the splitting (\ref{whitneyd}), with respective h- and v--projections, $%
2h=I+\ ^{\ast }\mathbf{N}$ and $2v=I-\ ^{\ast }\mathbf{N.}$ $\Box $
\end{proof}

\vskip5pt

We note that for Hamilton spaces a connection--pair structure $\ ^{\ast
}\phi $ is symmetric, i.e.\ $^{\ast }N_{ia}=\ ^{\ast }N_{ai}$ and $^{\ast
}g^{ab}=\ ^{\ast }g^{ba}.$ For simplicity, in this work, we shall restrict
our considerations only to $^{\ast }\phi $--symmetric configurations.

In order to perform geometric constructions adapted both to the
N--connection and almost symplectic structure, it is necessary to work with $%
^{\ast }\phi $--adapted bases (\ref{fder}) and (\ref{fdif}) instead of,
respectively, (\ref{cdder}) and (\ref{cddif}). For instance, every vector
field $\mathbf{X}$ has two components
\begin{equation*}
\mathbf{X}=hX+vX=h^{\prime }X+wX,
\end{equation*}%
where $wX=w(\mathbf{X}).$ This defines the class of $^{\ast }\phi $--tensor
fields alternatively to that of d--tensor fields considered in the previous
sections.

\begin{definition}
A linear connection $\ ^{\phi }\mathbf{D}$ on $T^{\ast }M$ is a $^{\ast
}\phi $--connection if
\begin{equation*}
\ ^{\phi }\mathbf{D\ }^{\ast }\phi =0\mbox{ and }\ ^{\phi }\mathbf{D\ }w=0.
\end{equation*}
\end{definition}

We can characterize $\ ^{\phi }\mathbf{D}=(\ _{h}^{\phi }D\mathbf{=\{}\
^{\phi }L_{jk}^{i}\},\ _{w}^{\phi }D\mathbf{=\{}\ ^{\phi }C_{\ i}^{j\ k}\})$
by the coefficients computed with respect to $^{\ast }\phi $--adapted frames
(\ref{fder}) and (\ref{fdif}):%
\begin{eqnarray*}
\ ^{\phi }\mathbf{D}_{\ ^{\phi }e_{k}}\ ^{\phi }e_{j} &=&\ ^{\phi
}L_{jk}^{i}\ ^{\phi }e_{j},\ \ ^{\phi }\mathbf{D}_{\ ^{\phi }e_{k}}\ \
^{w}e^{b}=\ ^{\phi }L_{ck}^{b}\ \ ^{w}e^{c}, \\
\ ^{\phi }\mathbf{D}_{\ \ ^{w}e^{a}}\ \ ^{w}e^{b} &=&\ ^{\phi }C_{\ c}^{b\
a}\ \ ^{w}e^{c},\ ^{\phi }\mathbf{D}_{\ \ ^{w}e^{a}}\ \ ^{\phi }e_{j}=\
^{\phi }C_{\ j}^{k\ a}\ \ ^{\phi }e_{k},
\end{eqnarray*}%
for $\ ^{\phi }L_{n+j\ k}^{n+i}=-\ ^{\phi }L_{jk}^{i}$ and $\ ^{\phi }C_{\
n+j}^{\ n+k\ a}=-\ ^{\phi }C_{\ j}^{k\ a}.$

\begin{theorem}
There is a canonical $^{\ast }\phi $--connection $\ ^{\phi }\widehat{\mathbf{%
D}}=(\ ^{\phi }\widehat{L}_{jk}^{i},\ ^{\phi }\widehat{C}_{\ j}^{k\ a})$ on $%
T^{\ast }M$ completely defined by a Hamilton fundamental function $H(x,p)$
and satisfying the conditions%
\begin{eqnarray*}
\ _{h}^{\phi }\widehat{D}(\ ^{\ast }g^{ab}) &\mathbf{=}&0\mbox{ and }\
_{v}^{\phi }\widehat{D}(\ ^{\ast }g^{ab})\mathbf{=}0, \\
\ ^{\phi }\widehat{T}_{\ ij}^{k} &=&0\mbox{ and }\ ^{\phi }\widehat{S}%
_{a}^{\ bc}=0.
\end{eqnarray*}
\end{theorem}

\begin{proof}
The torsion and curvature of $\ ^{\phi }\widehat{\mathbf{D}}$ are computed
as in Theorem \ref{tcsteq}, see formulas (\ref{tcncs}), but with respect to (%
\ref{fder}) and (\ref{fdif}). By straightforward computations, we can verify
that for
\begin{eqnarray}
\ ^{\phi }\widehat{L}_{jk}^{i} &=&\frac{1}{2}\ ^{\ast }g^{km}(\ ^{\phi
}e_{i}\ ^{\ast }g_{mj}+\ ^{\phi }e_{j}\ ^{\ast }g_{im}-\ ^{\phi }e_{m}\
^{\ast }g_{ij})  \notag \\
&&-\frac{1}{2}\ ^{\ast }g^{km}(A_{jim}+A_{ijm}+A_{mij}),  \label{ccpc} \\
\ ^{\phi }C_{\ i}^{j\ k} &=&-\frac{1}{2}\ ^{\ast }g_{im}(\ ^{w}e^{j}\ ^{\ast
}g^{mk}+\ ^{w}e^{k}\ ^{\ast }g^{im}+\ ^{w}e^{m}\ ^{\ast }g^{jk})  \notag \\
&&-\frac{1}{2}\ ^{\ast }g_{im}(B^{jmk}+B^{kmj}+B^{mkj}),  \notag
\end{eqnarray}%
where
\begin{equation*}
A_{ijm}=\ ^{\ast }\Omega _{ijm}\mbox{ and }B^{mkj}=\ ^{\ast }g^{mi}(\ ^{\ast
}g^{ka}\ ^{w}e^{j}\ ^{\ast }N_{ia}-\ ^{\ast }g^{ja}\ ^{w}e^{k}\ ^{\ast
}N_{ia}),
\end{equation*}%
the conditions of Theorem are satisfied. $\Box $
\end{proof}

\vskip5pt

The diffeomorphism symmetry of Hamilton mechanical models and possible lifts
of Einstein spaces on cotangent bundles is an important characteristic of
classical theories. In deformation quantization models, this property can be
preserved for regular N--connection structures:

\begin{definition}
A diffeomorphism $f\in Diff(T^{\ast }M)$ is called $\ ^{\ast }\mathbf{N}$%
--regular if the restriction of tangent map $\left( \pi f\right) _{\ast
}:hT^{\ast }M\rightarrow TM$ $\ $to $hT^{\ast }M$ is a diffeomorphism.
\end{definition}

For a connection--pair $^{\ast }\phi ,$ the concept of $\ ^{\ast }\mathbf{N}$%
--regularity imposes the equivalence of statements: 1) the push--forward map
of $^{\ast }\phi $ by $f,$ defined as $^{\ast }\overline{\phi }=f_{\ast }(\
^{\ast }\phi )f_{\ast }^{-1},$ is a connection--pair and 2) $f$ is $\ ^{\ast
}\mathbf{N}$--regular. This follows from the mapping $\pi _{\ast }:\ker (I-\
^{\ast }\overline{\phi })\rightarrow TM$ and the equality $f_{\ast }(\ker
(I-\ ^{\ast }\phi ))=\ker (I-\ ^{\ast }\overline{\phi }).$ The N--connection
$\ ^{\ast }\overline{\mathbf{N}}$ associated to $\ ^{\ast }\overline{\phi }$
is the push--forward of $\ ^{\ast }\mathbf{N}$ by $f.$ One says that $\
^{\ast }\phi $ and $\ ^{\ast }\overline{\phi }$ are $f$--related.

Following a calculus with local coordinate transforms, one proves:

\begin{corollary}
1) For $\ ^{\phi }\mathbf{e}_{\alpha }=(\ ^{\phi }e_{i},\ ^{w}e^{b})$ in $%
(x,p)$ and $\ ^{\phi }\overline{\mathbf{e}}_{\alpha }=(\ ^{\phi }\overline{e}%
_{i},\ ^{w}\overline{e}^{b})$ in $(\overline{x}(x,p),\overline{p}(x,p))$
induced respectively by $^{\ast }\phi $ and $\ ^{\ast }\overline{\phi },$
the following formulas hold%
\begin{equation*}
f_{\ast }(\ ^{\phi }e_{i})=\varpi _{i}^{k}\ ^{\phi }\overline{e}_{k}%
\mbox{
and }f_{\ast }(\ ^{w}e^{b})=\widetilde{\varpi }_{c}^{b}\ \ ^{w}e^{c},
\end{equation*}%
where
\begin{equation*}
\varpi _{i}^{k}\ ^{\ast }N_{ka}=\ ^{\phi }e_{i}(\overline{p}_{a})\mbox{ and }%
\widetilde{\varpi }_{c}^{b}=\ ^{w}e^{b}(\overline{p}_{c})-\ ^{w}e^{b}(%
\overline{x}^{k})\ ^{\ast }\overline{N}_{kc}.
\end{equation*}%
2) The push--forward of a N--connection $\ ^{\ast }\mathbf{N}$ by a $\
^{\ast }\mathbf{N}$--regular diffeomorphism is a N--connection if $f$ is
fiber preserving, i.e. locally $f(x,p)=$ $(\overline{x}(x),\overline{p}%
(x,p)).$
\end{corollary}

Symplectic morphisms are diffeomorphisms which transform a symplectic form $\
^{\ast }\theta =dp_{i}\wedge dx^{i}$ into a symplectic form $\ ^{\ast }%
\overline{\theta }=f^{\ast }(\ ^{\ast }\theta )=d\overline{p}_{i}\wedge d%
\overline{x}^{i}.$ By coordinate transforms, one proves:

\begin{theorem}
For a $f\in Diff(T^{\ast }M)$ being $\ ^{\ast }\mathbf{N}$--regular and
satisfying the condition $f^{\ast }(\ ^{\ast }\theta )(\ ^{\phi }e_{i},\
^{w}e^{b})=\delta _{i}^{b},$ two from the next statements implies the third:
1) $f$ is a symplectic morphism; 2) $^{\ast }\phi $ is symmetric; 3)$\ ^{\ast }%
\overline{\phi }$ is symmetric.
\end{theorem}

The push--forward of $\ ^{\ast }\mathbf{g}$ (\ref{dmctb}) results in the
local form%
\begin{equation*}
\ ^{\ast }\overline{\mathbf{g}}=\ ^{\ast }\overline{\mathbf{g}}_{\alpha
\beta }\ ^{\ast }\overline{\mathbf{e}}^{\alpha }\otimes \ ^{\ast }\overline{%
\mathbf{e}}^{\beta }=\ ^{\ast }\overline{g}_{ij}(x,p)\overline{e}^{i}\otimes
\overline{e}^{j}+\ ^{\ast }\overline{g}^{ab}(x,p)\ ^{\ast }\overline{\mathbf{%
e}}_{a}\otimes \ ^{\ast }\overline{\mathbf{e}}_{b},
\end{equation*}%
where $\ ^{\ast }\overline{g}^{ij}\circ f=\varpi _{k}^{i}\varpi _{l}^{j}\
^{\ast }g^{kl}.$

If $\nabla $ is a linear connection on $T^{\ast }M,$ we define its
push--forward by $f$ $\ $on $T^{\ast }M$ as
\begin{equation*}
\overline{\nabla }_{\overline{X}}\overline{Y}\doteqdot f_{\ast }\left(
\nabla _{X}Y\right)
\end{equation*}%
for $\overline{Y}=f_{\ast }(Y)$ and $\overline{X}=f_{\ast }(X).$
By coordinate parameterizations of diffeomorphisms, we can prove:

\begin{proposition}
1) A connection $\nabla $ is a $^{\ast }\phi $--connection if and only if $%
\overline{\nabla }$ is a $^{\ast }\overline{\phi }$--connection;
2) $\nabla $ is compatible to metric $\ ^{\ast }\mathbf{g}$
(almost symplectic $\ ^{\ast }\theta )$ structure if and only if \
$\overline{\nabla }$ is compatible to metric $\
^{\ast }\overline{\mathbf{g}}$ (almost symplectic $\ ^{\ast }\overline{%
\theta })$ structure.
\end{proposition}

\begin{proof}
Locally, we can prove that the coefficients of $\ ^{\ast }\phi $--connection
$^{\phi }\mathbf{D}$ are related to coefficients of $^{\ast }\overline{\phi }
$--connection $^{\phi }\overline{\mathbf{D}}$ by formulas
\begin{eqnarray*}
\ ^{\phi }\overline{L}_{ij}^{k} &=&\widetilde{\varpi }_{i}^{i^{\prime }}%
\widetilde{\varpi }_{j}^{j^{\prime }}\varpi _{k^{\prime }}^{k}\ ^{\phi
}L_{j^{\prime }k^{\prime }}^{i^{\prime }}+\varpi _{i^{\prime }}^{k}\ ^{\phi }%
\overline{e}_{j}\ ^{\phi }\overline{e}_{i}(x^{i^{\prime }}), \\
\ ^{\phi }\overline{C}_{\ k}^{i\ j} &=&\varpi _{i^{\prime }}^{i}\varpi
_{j^{\prime }}^{j}\widetilde{\varpi }_{k}^{k^{\prime }}\ ^{\phi }C_{\
k^{\prime }}^{i^{\prime }\ j^{\prime }}+\varpi _{i^{\prime }}^{i}\ ^{w}%
\overline{e}^{j}\ ^{\phi }\overline{e}_{i}(x^{i^{\prime }}).
\end{eqnarray*}%
Using formulas
\begin{eqnarray*}
\ ^{\phi }\overline{\mathbf{D}}(\ ^{\ast }\phi ) &=&0\Longleftrightarrow \
^{\phi }\mathbf{D(\ ^{\ast }\phi )}=0, \\
\ ^{\phi }\overline{\mathbf{D}}\ ^{\ast }\overline{\mathbf{g}} &=&\ ^{\phi }%
\overline{\mathbf{D}}\left[ f_{\ast }(\ ^{\ast }\mathbf{g})\right] =f_{\ast
}(\ ^{\phi }\mathbf{D}\ ^{\ast }\mathbf{g}), \\
\ ^{\phi }\overline{\mathbf{D}}\ ^{\ast }\overline{\mathbf{\theta }} &=&\
^{\phi }\overline{\mathbf{D}}\left[ f_{\ast }(\ ^{\ast }\mathbf{\theta })%
\right] =f_{\ast }(\ ^{\phi }\mathbf{D}\ ^{\ast }\mathbf{\theta }),
\end{eqnarray*}%
where for $^{\ast }\phi $--connections $\ $we use the symbol $^{\phi }%
\mathbf{D}$ instead of $\nabla ,$ \ and by local computations, we can verify
that there are satisfied the conditions%
\begin{eqnarray*}
\ ^{\phi }\overline{\mathbf{D}}\ ^{\ast }\overline{\mathbf{g}}
&=&0\Longleftrightarrow \ ^{\phi }\mathbf{D}\ ^{\ast }\mathbf{g}=0, \\
\ ^{\phi }\overline{\mathbf{D}}\ ^{\ast }\overline{\mathbf{\theta }}
&=&0\Longleftrightarrow \ ^{\phi }\mathbf{D}\ ^{\ast }\mathbf{\theta }=0%
\mathbf{.}
\end{eqnarray*}%
$\Box $
\end{proof}

\vskip5pt

Summarizing the above--presented results we get the proof of:

\begin{theorem}
There is a unique canonical $^{\ast }\phi $--connection $^{\phi }\widehat{%
\mathbf{D}}$ with coefficients (\ref{ccpc}) generated by a Hamilton
fundamental function $H(x,p)$ which is almost symplectic and metric
compatible and adapted to symplectic morphism.
\end{theorem}

We can use $^{\phi }\widehat{\mathbf{D}}$ and frames (\ref{fder}) and (\ref%
{fdif}) to construct instead of (\ref{cdcop}) an extended on $\ \mathbf{%
^{\ast }}\mathcal{W}\otimes \Lambda $ operator
\begin{equation}
\ ^{\phi }\widehat{\mathbf{D}}\left( a\otimes \lambda \right) \doteqdot
\left( \ ^{\phi }\mathbf{e}_{\alpha }(a)-u^{\beta }\ \ ^{\phi }\widehat{%
\mathbf{\Gamma }}\mathbf{_{\alpha \beta }^{\gamma }\ }^{z\phi }\mathbf{e}%
_{\alpha }(a)\right) \otimes (\ ^{\phi }\mathbf{e}^{\alpha }\wedge \lambda
)+a\otimes d\lambda ,  \label{secextop}
\end{equation}%
where the local basis $^{z\phi }\mathbf{e}_{\alpha }$ is $\ ^{\phi }\mathbf{e%
}_{\alpha }$ redefined in $z$--variables. This allows us to introduce a new
class of operators adapted both to the N--connection structure and
symplectic morphisms:

\begin{definition}
The Fedosov--Hamilton operator--pairs $\ ^{\phi }\mathbf{\delta }$ and $\
^{\phi }\mathbf{\delta }^{-1}$ on \newline
$\ \ ^{\ast }\mathcal{W}\otimes \mathbf{\Lambda ,}$ are defined%
\begin{equation}
\ ^{\phi }\mathbf{\delta }(a)=\ ^{\phi }\mathbf{e}^{\alpha }\wedge \mathbf{\
}^{z\phi }\mathbf{e}_{\alpha }(a),\ \mbox{and\ }\ \ ^{\phi }\mathbf{\delta }%
^{-1}(a)=\left\{
\begin{array}{c}
\frac{i}{p+q}z^{\alpha }\ \ ^{\phi }\mathbf{e}_{\alpha }(a),\mbox{ if }p+q>0,
\\
{\qquad 0},\mbox{ if }p=q=0,%
\end{array}%
\right.  \label{fedopp}
\end{equation}%
where any $a\in \ \mathbf{^{\ast }}\mathcal{W}\otimes \mathbf{\Lambda }$ is
homogeneous w.r.t. the grading $\deg _{s}$ and $\deg _{a}$ with $\deg
_{s}(a)=p$ and $\deg _{a}(a)=q.$
\end{definition}

We note that the formulas (\ref{fedopp}) are different from (\ref{feddop})
because they are defined for different nonholonomic distributions and
related adapted frame structures.

Using differential calculus of forms on $\mathbf{^{\ast }}\mathcal{W}%
\otimes \mathbf{\Lambda ,}$ we prove

\begin{proposition}
\label{prthprfo}The torsion and curvature canonical d--operators of\newline
$\ ^{\phi }\widehat{\mathbf{D}}\left( a\otimes \lambda \right) $ (\ref%
{secextop}) are computed
\begin{equation}
^{z\phi }\widehat{\mathcal{T}}\ \doteqdot \frac{z^{\gamma }}{2}\ \ ^{\ast
}\theta _{\gamma \tau }\ \ ^{\phi }\widehat{\mathbf{T}}_{\alpha \beta
}^{\tau }(u)\ ^{\phi }\mathbf{e}^{\alpha }\wedge \ ^{\phi }\mathbf{e}^{\beta
},  \label{at1}
\end{equation}%
and%
\begin{equation}
\ ^{z\phi }\widehat{\mathcal{R}}\doteqdot \frac{z^{\gamma }z^{\varphi }}{4}\
^{\ast }\theta _{\gamma \tau }\ \ ^{\phi }\widehat{\mathbf{R}}_{\ \varphi
\alpha \beta }^{\tau }(u)\ \ ^{\phi }\mathbf{e}^{\alpha }\wedge \ ^{\phi }%
\mathbf{e}^{\beta },  \label{ac1}
\end{equation}%
where the nontrivial coefficients of $\ \ ^{\phi }\widehat{\mathbf{T}}%
_{\alpha \beta }^{\tau }$ and $\ \ ^{\phi }\widehat{\mathbf{R}}_{\ \varphi
\alpha \beta }^{\tau }$ are defined respectively by formulas (\ref{tcncs})
and (\ref{ccncs}) for the canonical $^{\ast }\phi $--connection $\ ^{\phi }%
\widehat{\mathbf{D}}=(\ ^{\phi }\widehat{L}_{jk}^{i},\ ^{\phi }\widehat{C}%
_{\ j}^{k\ a})$ (\ref{ccpc}) with respect to (\ref{fder}) and (\ref{fdif}).
\end{proposition}

By straightforward verifications, one gets the proof of

\begin{theorem}
\label{thprfo} Any Fedosov--Hamilton operator--pairs (\ref{fedopp}) and
related extended ope\-rator--pair (\ref{secextop}) are defined by torsion (%
\ref{at1}) and curvature (\ref{ac1}) following formulas:
\begin{equation}
\left[ \ ^{\phi }\widehat{\mathbf{D}},\ ^{\phi }\mathbf{\delta }\right] =%
\frac{i}{v}ad_{Wick}(^{z\phi }\widehat{\mathcal{T}})\mbox{ and }\ \ ^{\phi }%
\widehat{\mathbf{D}}^{2}=-\frac{i}{v}ad_{Wick}(\ ^{z\phi }\widehat{\mathcal{R%
}}),  \label{ffedop}
\end{equation}%
where $[\cdot ,\cdot ]$ is the $\deg _{a}$--graded commutator of
endomorphisms of $\ \mathbf{^{\ast }}\mathcal{W}\otimes \mathbf{\Lambda }$
and $ad_{Wick}$ is defined via the $\deg _{a}$--graded commutator in $\left(
\ \mathbf{^{\ast }}\mathcal{W}\otimes \mathbf{\Lambda ,\circ }\right) .$
\end{theorem}

The formulas (\ref{ffedop}) can be restated for any metric
compatible d--connecti\-on and $^{\ast }\phi $--connection
structures on $T^{\ast }M.$ Finally, we conclude that the
operators constructed in this section are invariant under
diffeomorphisms and in particular under symplectic morphisms.

\section{Deformation Quantization of Hamilton and Einstein Spaces}

The aim of this section is to provide the main Fedosov's theorems for
Hamilton spaces and show how the Einstein manifolds can be encoded into the
topological structure of such quantized nonholonomic spaces.

\subsection{Fedosov's theorems for connection--pairs}

The theorems will be formulated for the canonical $^{\ast }\phi $%
--connection $\ ^{\phi }\widehat{\mathbf{D}}$ (\ref{ccpc}).

\begin{theorem}
\label{th3a}A Hamilton fundamental function $H(x,p)$ defines a flat Fedosov $%
^{\ast }\phi $--connection
\begin{equation*}
\ \ ^{\phi }\widehat{\mathcal{D}}\doteqdot -\ ^{\phi }\mathbf{\delta }+\
^{\phi }\widehat{\mathbf{D}}-\frac{i}{v}ad_{Wick}(\ ^{\phi }r)
\end{equation*}%
satisfying the condition $\ ^{\phi }\widehat{\mathcal{D}}^{2}=0,$ where the
unique element $\ ^{\phi }r\in $ $\ ^{\ast }\mathcal{W}\otimes \mathbf{%
\Lambda ,}$ $\deg _{a}(\ ^{\phi }r)=1,$ $\ ^{\phi }\mathbf{\delta }^{-1}\
^{\phi }r=0,$ is a solution of equation
\begin{equation*}
\ ^{\phi }\mathbf{\delta }\ ^{\phi }r=\ ^{\phi }\widehat{\mathcal{T}}\ +\
^{\phi }\widehat{\mathcal{R}}+\ ^{\phi }\widehat{\mathbf{D}}\ ^{\phi }r-%
\frac{i}{v}\ ^{\phi }r\circ \ ^{\phi }r.
\end{equation*}%
The solution for $\ ^{\phi }r$ can be computed recursively with respect to
the total degree $Deg$ in the form%
\begin{eqnarray*}
\ ^{\phi }r^{(0)} &=&\ ^{\phi }r^{(1)}=0,\ ^{\phi }r^{(2)}=\ ^{\phi }\mathbf{%
\delta }^{-1}\ ^{\phi }\widehat{\mathcal{T}}, \\
\ ^{\phi }r^{(3)} &=&\ ^{\phi }\mathbf{\delta }^{-1}\left( \ ^{\phi }%
\widehat{\mathcal{R}}+\ ^{\phi }\widehat{\mathbf{D}}\ ^{\phi }r^{(2)}-\frac{i%
}{v}\ ^{\phi }r^{(2)}\circ \ ^{\phi }r^{(2)}\right) , \\
\ ^{\phi }r^{(k+3)} &=&\ ^{\phi }\mathbf{\delta }^{-1}\left( \ ^{\phi }%
\widehat{\mathbf{D}}\ ^{\phi }r^{(k+2)}-\frac{i}{v}\sum\limits_{l=0}^{k}\
^{\phi }r^{(l+2)}\circ \ ^{\phi }r^{(l+2)}\right) ,k\geq 1,
\end{eqnarray*}%
where $a^{(k)}$ is the $Deg$--homogeneous component of degree $k$ of an
element $a\in $ $\ \ ^{\ast }\mathcal{W}\otimes \mathbf{\Lambda }.$
\end{theorem}

\begin{proof}
It follows from a local component calculus with $^{\ast }\phi $--adapted
coefficients of $\ ^{\phi }\widehat{\mathbf{D}}$, see similar considerations
in \cite{fed1,fed2,karabeg1}.$\square $
\end{proof}

\vskip5pt

For Hamilton spaces, we can define a canonical star--product. By an explicit
construction, we prove:

\begin{theorem}
\label{th3b}A $^{\ast }\phi $--adapted star--product for Hamilton spaces is
defined on $C^{\infty }(T^{\ast }M)[[v]]$ by formula
\begin{equation*}
\ ^{1}f\ast \ ^{2}f\doteqdot \sigma (\tau (\ ^{1}f))\circ \sigma (\tau (\
^{2}f)),
\end{equation*}%
where the projection $\sigma :\ ^{\ast }\mathcal{W}\rightarrow C^{\infty
}(T^{\ast }M)[[v]]$ onto the part of $\deg _{s}$--degree zero is a bijection
and the inverse map $\tau :C^{\infty }(T^{\ast }M)[[v]]\rightarrow \ ^{\ast }%
\mathcal{W}$ $\ $can be calculated recursively w.r..t the total degree $Deg,$%
\begin{eqnarray*}
\tau (f)^{(0)} &=&f\mbox{\ and, for \ }k\geq 0, \\
\tau (f)^{(k+1)} &=&\ \check{\delta}^{-1}\left( \ ^{\phi }\widehat{\mathbf{D}%
}\tau (f)^{(k)}-\frac{i}{v}\sum\limits_{l=0}^{k}ad_{Wick}(\ ^{\phi
}r^{(l+2)})(\tau (f)^{(k-l)})\right) .
\end{eqnarray*}
\end{theorem}

Let $\ ^{f}X$ the Hamiltonian vector field corresponding to a function $f\in
C^{\infty }(T^{\ast }M)$ on space $(T^{\ast }M,\ ^{\ast }\theta )$ and
consider the antisymmetric part
\begin{equation*}
\ ^{-}C(\ ^{1}f,\ ^{2}f)\ \doteqdot \frac{1}{2}\left( C(\ ^{1}f,\ ^{2}f)-C(\
^{2}f,\ ^{1}f)\right)
\end{equation*}
of bilinear operator $C(\ ^{1}f,\ ^{2}f).$ A star--product
(\ref{starp}) is normalized if $\ _{1}C(\ ^{1}f,$ $\
^{2}f)=\frac{i}{2}\{\ ^{1}f,\ ^{2}f\},$ where $\{\cdot ,\cdot \}$
is the Poisson bracket, see (\ref{poisbr}). For the normalized
$\ast ,$ the bilinear operator $\ _{2}^{-}C$ defines a de
Rham--Chevalley 2--cocycle, when there is a unique closed 2--form
$\ \
^{\phi }\varkappa $ such that%
\begin{equation}
\ _{2}C(\ ^{1}f,\ ^{2}f)=\frac{1}{2}\ \ ^{\phi }\varkappa (\ ^{f_{1}}X,\
^{f_{2}}X)  \label{c2}
\end{equation}%
for all $\ ^{1}f,\ ^{2}f\in C^{\infty }(T^{\ast }M).$ This is used to
introduce $c_{0}(\ast )\doteqdot \lbrack \ ^{\phi }\varkappa ]$ as the
equivalence class.

Computing $\ _{2}C$ from (\ref{c2}) and using the Theorem \ref{th3b}, we get
the proof for

\begin{lemma}
\label{lem1}The unique 2--form defined by the unique canonical
connection-\--pair can be computed
\begin{equation*}
\ ^{\phi }\varkappa =-\frac{i}{8}\ ^{\ast }\mathbf{J}_{\tau }^{\ \alpha
^{\prime }}\ ^{\phi }\widehat{\mathcal{R}}_{\ \alpha ^{\prime }}^{\tau }-%
\frac{i}{6}d\left( \ ^{\ast }\mathbf{J}_{\tau }^{\ \alpha ^{\prime }}\
^{\phi }\widehat{\mathbf{T}}_{\ \alpha ^{\prime }\beta }^{\tau }\ \ ^{\phi }%
\mathbf{e}^{\beta }\right) .
\end{equation*}
\end{lemma}

Let us define the canonical class $\ ^{\phi }\varepsilon ,$ for the
splitting (\ref{cpspl}). We can perform a distinguished complexification of
such second order tangent bundles in the form $T_{\mathbb{C}}\left( \ ^{\phi
}TT^{\ast }M\right) =T_{\mathbb{C}}\left( hT^{\ast }M\right) \oplus T_{%
\mathbb{C}}\left( wT^{\ast }M\right) $ and introduce $\ ^{\phi }\varepsilon $
as the first Chern class of the distributions $T_{\mathbb{C}}^{\prime
}\left( \ ^{\phi }TT^{\ast }M\right) =T_{\mathbb{C}}^{\prime }\left(
hT^{\ast }M\right) \oplus T_{\mathbb{C}}^{\prime }\left( wT^{\ast }M\right) $
of couples of vectors of type $(1,0)$ both for the h-- and w--parts. We
compute $\ ^{\phi }\varepsilon $ using $\ ^{\phi }\widehat{\mathbf{D}}$ and
the h- and v--projections $h\Pi =\frac{1}{2}(Id_{h}-iJ_{h})$ and $v\Pi =%
\frac{1}{2}(Id_{v}-iJ_{v}),$ where $Id_{h}$ and $Id_{v}$ are respective
identity operators and $J_{h}$ and $J_{v}$ are almost complex operators,
which are projection operators onto corresponding $(1,0)$--subspaces.
Introducing the  matrix $\left( h\Pi ,v\Pi \right) \ \widehat{\mathcal{R}}%
\left( h\Pi ,v\Pi \right) ^{T},$ where $(...)^{T}$ means transposition, as
the curvature matrix of the N--adapted restriction of $\ ^{\phi }\widehat{%
\mathbf{D}}$ to $T_{\mathbb{C}}^{\prime }\left( \ ^{\phi }TT^{\ast }M\right)
,$ we compute the closed Chern--Weyl form
\begin{eqnarray}
\ ^{\phi }\gamma &=&-iTr\left[ \left( h\Pi ,v\Pi \right) \ ^{\phi }\widehat{%
\mathcal{R}}\left( h\Pi ,v\Pi \right) ^{T}\right]  \label{aux4} \\
&=&-iTr\left[ \left( h\Pi ,v\Pi \right) \ ^{\phi }\widehat{\mathcal{R}}%
\right] =-\frac{1}{4}\ ^{\ast }\mathbf{J}_{\tau }^{\ \alpha ^{\prime }}\
^{\phi }\widehat{\mathcal{R}}_{\ \alpha ^{\prime }}^{\tau }.  \notag
\end{eqnarray}%
Using the canonical class $\ ^{\phi }\varepsilon \doteqdot \lbrack \ ^{\phi
}\gamma ],$  we prove:

\begin{theorem}
\label{th3c}The zero--degree cohomology coefficient $c_{0}(\ast )$ for the
almost K\"{a}hler model of a Hamilton space$\mathbf{\ \ ^{\ast }K}^{2n},$
defined by a Hamilton fundamental function $H(x,p)$ is computed $c_{0}(\ast
)=-(1/2i)\ ^{\phi }\varepsilon .$
\end{theorem}

The coefficient $c_{0}(\ast )$ contains as a particular case the class of
zero--degree cohomologies computed for a metric of type $g_{i^{\prime
}j^{\prime }}(x)$ on $M,$ defining a solution of the Einstein equations and
lifted on cotangent bundle by formula (\ref{eq01}). In such cases, this
zero--degree coefficient defines certain quantum properties of the
gravitational field. A more rich geometric structure should be considered if
we define a value similar to $c_{0}(\ast )$ encoding the information about
Einstein manifolds lifted to the cotangent bundle.

\subsection{Quantum gravitational field equations}

Any solution in classical Einstein gravity can be embedded into a Hamilton
space model and quantized on cotangent bundle following the Fedosov
quantization adapted to $^{\ast }\phi $-- and d--connections. Considering a
de--quantization formalism \cite{karabeg2}, we construct certain quantum
deformations of the classical Einstein configurations in the classical
limit. Such a model defines a nonholonomic almost K\"{a}hler generalization
of the Einstein gravity on cotangent bundle. The solutions for the
''cotangent'' gravity are, in general, with violation of Lorentz symmetry
induced by quantum corrections. The nature of such quantum gravity
corrections is different from those defined by Finsler--Lagrange models on
tangent bundle (see, for instance, \cite{ma1987,ma,glsf,mign,ggp})), locally
anisotropic string gravity \cite{vstrf,vncsup,mavr} with corrections from
extra--dimensions and nonholonomic spinor gravity \cite{vstav,vfs,vhs} and
noncommutative gravity, see reviews of results in \cite{vsgg,vrfg}. The aim
of this section is to analyze how a generalization of Einstein gravity can
be performed on cotangent bundles in terms of canonical $^{\ast }\phi $%
--connections, with geometric structures induced by an effective Hamiltonian
fundamental function, when the Fedosov quantization can be naturally
performed.

For a canonical $^{\ast }\phi $--connections $\ ^{\phi }\widehat{\mathbf{D}}%
\mathbf{=\{\ ^{\phi }\Gamma \}}$ (\ref{ccpc}), we can define the Ricci
tensor,
\begin{equation*}
Ric(\ ^{\phi }\widehat{\mathbf{D}})=\{\ ^{\phi }\widehat{\mathbf{R}}_{\
\beta \gamma }\doteqdot \ ^{\phi }\widehat{\mathbf{R}}_{\ \beta \gamma
\alpha }^{\alpha }\},
\end{equation*}%
and the scalar curvature, $\ ^{\phi }R\doteqdot \ ^{\ast }\mathbf{g}^{\alpha
\beta }\ ^{\phi }\widehat{\mathbf{R}}_{\alpha \beta }.$ On cotangent bundle $%
T^{\ast }M,$ we postulate the field equations
\begin{equation}
\ ^{\phi }\widehat{\mathbf{R}}_{\ \beta }^{\underline{\alpha }}-\frac{1}{2}%
(\ ^{\phi }R+\lambda )\ ^{\phi }\mathbf{e}_{\ \beta }^{\underline{\alpha }%
}=8\pi G\mathbf{\Upsilon }_{\ \beta }^{\underline{\alpha }},
\label{deinsteq}
\end{equation}%
where $\mathbf{\Upsilon }_{\ \beta }^{\underline{\alpha }}$ is the effective
energy--momentum tensor, $\lambda $ is the cosmological constant, $G$ is the
Newton constant in the units when the light velocity $c=1,$ and $\ ^{\phi }%
\mathbf{e}_{\ \beta }=\ ^{\phi }\mathbf{e}_{\ \beta }^{\underline{\alpha }%
}\partial /\partial u^{\underline{\alpha }}$ is the $^{\ast }\phi $--adapted
base (\ref{fder}).

We consider the effective source 3--form
\begin{equation*}
\overleftarrow{\Upsilon }_{\ \beta }=\mathbf{\Upsilon }_{\ \beta }^{%
\underline{\alpha }}\ \epsilon _{\underline{\alpha }\underline{\beta }%
\underline{\gamma }\underline{\delta }}du^{\underline{\beta }}\wedge du^{%
\underline{\gamma }}\wedge du^{\underline{\delta }},
\end{equation*}%
where $\epsilon _{\alpha \beta \gamma \delta }$ is the absolute
antisymmetric tensor, and the action for the ''cotangent'' gravity and
matter fields,
\begin{equation*}
S[\mathbf{e,\Gamma ,\phi }]=\ ^{gr}S[\mathbf{e,\Gamma }]+\ ^{matter}S[%
\mathbf{e,\Gamma ,\phi }].
\end{equation*}

\begin{theorem}
\label{theq}The equations (\ref{deinsteq}) can be represented as 3--form
equations%
\begin{equation}
\epsilon _{\alpha \beta \gamma \tau }\left( \ ^{\phi }\mathbf{e}^{\alpha
}\wedge \ ^{\phi }\mathcal{R}^{\beta \gamma }+\lambda \ ^{\phi }\mathbf{e}%
^{\alpha }\wedge \ ^{\phi }\mathbf{e}^{\beta }\wedge \ ^{\phi }\mathbf{e}%
^{\gamma }\right) =8\pi G\overleftarrow{\Upsilon }_{\ \tau }  \label{einsteq}
\end{equation}%
following from action by varying the components of $\mathbf{e}_{\ \beta },$
when%
\begin{equation*}
\overleftarrow{\Upsilon }_{\ \tau }=\ ^{m}\overleftarrow{\Upsilon }_{\ \tau }%
\mbox{ and }\ ^{m}\overleftarrow{\Upsilon }_{\ \tau }=\ ^{m}\mathbf{\Upsilon
}_{\ \tau }^{\underline{\alpha }}\epsilon _{\underline{\alpha }\underline{%
\beta }\underline{\gamma }\underline{\delta }}du^{\underline{\beta }}\wedge
du^{\underline{\gamma }}\wedge du^{\underline{\delta }},
\end{equation*}%
where $\ ^{m}\mathbf{\Upsilon }_{\ \tau }^{\underline{\alpha }}=\delta \
^{matter}S/\delta \mathbf{e}_{\underline{\alpha }}^{\ \tau }.$
\end{theorem}

\begin{proof}
It is a usual textbook and/or differential form calculus (see, for instance, %
\cite{mtw,rovelli}). In our case, we have to use the $^{\ast }\phi $--adapted bases (%
\ref{fder}) and (\ref{fdif}) for $\ ^{\phi }\widehat{\mathbf{D}}.$ $\square $
\end{proof}

\vskip3pt

The Chern--Weyl 2--form (\ref{aux4}) can be used to define the quantum
version\footnote{in the sense of deformation quantization
(i.e. when quantum equations are derived following a  deformation procedure)
 but not of perturbative quantum theory}
of Einstein equations (\ref{einsteq}) in the approaches with
deformation quantization:

\begin{corollary}
The quantum field equations on cotangent bundle generalizing the Einstein
gravitations in general relativity are
\begin{equation}
\ ^{\phi }\mathbf{e}^{\alpha }\wedge \ ^{\phi }\gamma =\epsilon ^{\alpha
\beta \gamma \tau }2\pi G\ ^{\ast }\mathbf{J}_{\beta \gamma }\overleftarrow{%
\Upsilon }_{\ \tau }\ -\frac{\lambda }{4}\ ^{\ast }\mathbf{J}_{\beta \gamma
}\ ^{\phi }\mathbf{e}^{\alpha }\wedge \ ^{\phi }\mathbf{e}^{\beta }\wedge \
^{\phi }\mathbf{e}^{\gamma }.  \label{aseq}
\end{equation}
\end{corollary}

\begin{proof}
Multiplying $\ ^{\phi }\mathbf{e}^{\alpha }\wedge $ with (\ref{aux4}),
taking into consideration the equation (\ref{einsteq}), and introducing the
almost complex operator $^{\ast }\mathbf{J}_{\beta \gamma },$ we get the
almost symplectic form of the Einstein equations (\ref{aseq}). $\square $
\end{proof}

\vskip5pt

An explicit computation of $\ ^{\phi }\gamma $ for nontrivial matter fields
has to be performed for a deformation quantization model with interacting
gravitational and matter fields geometrized in terms of an almost K\"{a}hler
model defined for spinor and fiber bundles on cotangent bundles.

\section{Conclusions}

In this paper we outlined a method of converting any regular Lagrange and
Hamilton dynamics into equivalent almost K\"ahler geometries with canonical
nonlinear connection (N--connection) and adapted almost symplectic
structures. The formalism was performed to be invariant under
symplectic morphisms and adapted to Legendre transforms of Lagrangians into
Hamiltonians and inversely. The geometry of cotangent bundles endowed with
nonholonomic distributions, formulated as certain Hamilton--Cartan spaces
being dual to the corresponding Lagrange--Finsler spaces, described in
Sections 2 and 3 presents a key prerequisite of this approach to deformation
quantization.

Given a regular Lagrange (in particular, Finsler; and related Legendre
transforms), or Hamilton (in particular, Cartan), generating function, we
completely define the fundamental geometric objects of a Hamilton geometry
modelled on cotangent bundle, inducing canonical almost symplectic and
compatible symplectic connection.
The connection--pair and canonical $^{\ast}\phi $--connection uniquely constructed to preserve the invariance under
symplectic morphisms are introduced in Section 4.2.
Such geometric objects are crucial for defining Fedosov--Hamilton operator--pairs. This allows us to
generalize the Fedosov's theorems for deformation quantization to the case
of Hamilton spaces and to show how the Einstein manifolds can be encoded
into such a quantization scheme.

Any classical solution of the Einstein equations can be lifted to the
cotangent bundle, and embedded into a Hamilton space geometry using frame
transforms variables depending both on spacetime and phase (momentum like)
coordinates. Such variables can be defined from a de--quantization procedure
like that considered in Ref. \cite{karabeg2} but in our case re--formulated
for nonholonomic cotangent bundles. In general, such constructions result in
quasi--classical effects of quantum gravity with violations of local Lorentz
symmetry. The surprising result advocated in this paper is that we can model
classical and quantum gravitational effects by corresponding effective
classical and quantum Hamilton mechanics systems. Nevertheless, certain
additional phenomenological and/or experimental data for quantum gravity
effects have to be assigned to the scheme in order to define the
nonholonomic frame transforms of locally isotropic gravitational fields into
quantized and semi--classical ones on cotangent bundle.

On cotangent bundles and curved phase spaces, there were developed different
methods of quantization of nonlinear field theories and mechanical systems
with nonholonomic constraints. For instance, in Refs. \cite{bnw1,bnw2},
there were constructed in explicit form examples of "cotangent"
star--products (in such cases, one derive certain compatible symplectic /
Levi Civita type connections). The formalism was revised and developed by
introducing auxiliary variables, with further restrictions, and/or higher
spin systems, in analogy with gauge theories with generalized symmetries, in
the framework of the BRST approach \cite{bgst,vasil,grig,lyakh2,bff}. In
general, the methods of quantization of nonlinear physical systems on
contangent bundles are very different from that on tangent bundles. The
source of this is in the fact that Lagrange and Hamilton (both classical and
quantum) schemes may result in very different quantum models for nonlinear
classical theories: On phase spaces, we have to consider Legendre transform
and additional symmetries related to symplectic morphisms
(i.e. morphisms preserving a symplectic structure into a symplectic structure). In such a case, the
author of \cite{karabeg} had to work with symplectic groupoids and introduce
contravariant connections which modified substantially the Fedosov scheme of
quantization.

There are various ideas and approaches to quantize gravity
theories, including deformation quantization. In our partner works
\cite{vqgr3,vqgr4} we elaborated a direction related to effective
Lagrange--Finsler geometries by performing nonholonomic
deformations to quantum versions of (semi) Riemannian manifolds
preserving the local Lorentz invariance at least for
semi--classical approximations. Following the same geometric
methods developed for Finsler spaces, but extended for
nonholonomic manifolds, we proved that there are similar
quantization schemes of gravity and Lagrange--Finsler spaces
modelled on tangent bundles \cite{vqgr1,vqgr2}. Such constructions
also result in violations of local Lorentz symmetry but with
physical effects which are very different from those for models on
cotangent bundles. The general conclusion of this paper is that a
deformation quantization scheme for Hamilton spaces and related
generalizations of the Einstein gravity on cotangent bundle
results in more rich geometric structures and requests more
advanced geometric methods with nonlinear connections and
connection--pairs. We can elaborate a standard Fedosov formalism
also on cotangent bundle (not involving groupoid structures) by
introducing canonical nonlinear connections structures and
generalizing the concept of linear and distinguished connections
to that of connection pairs.

We are planing to compare different approaches to deformation quantization
of gravity (preserving or violating the local Lorentz invariance) and other
quantization schemes in our further works.

\vskip3pt

\textbf{Acknowledgement: }M. A. was partially supported by grant
CNCSIS,1158/2007, Romania. S. V. performed his work as a visitor at Fields
Institute.

\end{document}